\documentclass{amsart}

\usepackage{lmodern}
\usepackage[a4paper]{geometry}
\usepackage[english]{babel}
\usepackage{graphicx}
\usepackage{onlyamsmath}
\usepackage{amsfonts}
\usepackage{amssymb}
\usepackage{color}

\usepackage{array}
\usepackage{hyperref}
\usepackage{enumerate}

\newcommand{\KK}{{\mathcal K}}
\newcommand{\SSS}{{\mathcal S}}
\newcommand{\DD}{{\mathcal D}}

\newcommand{\E}{{ \mathbb E}}

\newcommand{\R}{{\mathbb R}}
\newcommand{\C}{{\mathbb C}}
\newcommand{\N}{{\mathbb N}}

\newcommand{\ep}{\varepsilon}
\newcommand{\Pro}{{\mathbb P}}

\newcommand{\indiq}{\hbox{\rm 1}{\hskip -2.8 pt}\hbox{\rm I}}

\newcommand{\ds}{\displaystyle}
\newcommand{\la}{\langle}
\newcommand{\ra}{\rangle}

\newcommand{\intP}{{\int_0^t \!\!\! \int_{-R}^R \! \int_\R}}
\newcommand{\intPh}{{\int_t^{t+h} \!\!\! \int_{-R}^R \! \int_\R}}
\newcommand{\intPst}{{ \int_{-R}^R \! \int_\R}}
\newcommand{\inab}{{(i \nabla)}}
\newcommand{\trho}{{\tilde \rho}}

\DeclareMathOperator{\Tr}{Tr}

\DeclareMathOperator{\erfc}{erfc}

\newtheorem{thm}{Theorem}
\newtheorem{prop}{Proposition}
\newtheorem{lem}{Lemma}
\newtheorem{rmk}{Remark}

\title{Rigorous derivation of Lindblad equations from quantum jumps processes in 1D}

\author{Christophe Gomez}
\address{Institut de Math\'ematiques de Marseille \\
  Universit\'e d'Aix-Marseille \& CNRS UMR 7373 \& \'Ecole Centrale Marseille.}
\email[C.~Gomez]{christophe.gomez@univ-amu.fr}
\author{Maxime Hauray}
\email[M.~Hauray]{maxime.hauray@univ-amu.fr}

\begin{document}

\begin{abstract}
We are interested by the behaviour of a 1D single heavy particle, interacting with an environment made of very fast particles in a thermal state. Assuming that the interactions are instantaneous, we construct an appropriate quantum jump process for the density operator of the heavy particle. In a weak-coupling limit (many interactions with few effect), we show that the solutions of jump process converge in law in the appropriate space towards the solution of a Lindbald master equation. To the best of our knowledge, it seems to be the first rigorous derivation of a dissipative quantum evolution equation.
\end{abstract}



\maketitle


\section{Introduction} \label{sec:intro}

There is a huge phyisical literature on the so-called ``decoherence induced by the interaction with the environment'', starting form the seminal work of Joos and Zeh~\cite{JoosZeh} in the '80 (some authors also quote the forgotten work of Mott~\cite{Mott} fifty years earlier as the first work on the subject). Many works followed, we may quote for instance~\cite{Diosi,AMS,DoddHall,HornSipe,Adler}, in which some dissipative master equations, also called Kossakowski-Lindblad equations~\cite{Kossak,Lindblad} are formally derived for  the density operator of an open system (a particle is enough) interacting with the environment.  

However this subject has not been studied from the mathematical point of view, to the best of our knowledge. Here, we present a rigorous derivation of a Lindblad equation, in a kind of ``weak coupling'' limit of the interaction with a simple environment: a thermal bath, described with the help of a Poisson Point Process (PPP). At random exponential time, our system (a particle) interact with a particle of the environment (a Gaussian Wave Packet with random position, momentum and spread), in a instantaneous way. We use for that super-operators (operator on density matrix) that model this kind of instantaneous collisions and where obtained by the second author and his collaborators in~\cite{AHN}. Of course, this is rather a toy model than a physical environment, but its simplicity allows us to derive in a ``weak coupling'' limit (many collisions with small effect) a Lindblad equation.  An interesting properties of our limit equation is that asides its dissipative part, it contains also a stochastic potential term, that comes together with the dissipative term, and lower in some sense the decoherence effect.


\medskip
\paragraph{ \bf The functional spaces.}

We consider a ``heavy'' particle lying in $\R$, interacting with ``light'' particles of the environment. We are interested only by the behavior of the heavy particle, but since it is an \emph{open system}, we will study its density operator $\rho$.  The natural  Hilbert space associated to the heavy particle is 
\[H := L^2(\R).\]
We recall that a state or density operator associated to that heavy particle is a compact operator on $H$, which is symmetric and nonnegative always in the space trace class space $\SSS_1$, defined by
\[
\SSS_1 :=  \bigl\{   \rho \in L(H,H), \text{ s.t.} \quad \Tr |\rho| < + \infty 
\bigr\},
\]
where $| \rho|$ stands for $\bigl( \rho^\ast \rho \bigr)^{\frac12}$. The notation $\SSS_1^+$ will stand for the subspace of $\SSS_1$ made of symmetric non-negative operator.
We will always assume that the initial density operator of the heavy particle  belongs to $\SSS^+_1$ and is normalized: $ \| \rho_0 \|_{\SSS_1} = \Tr \rho_0 =1$. 

It is well-known that trace class operator are also Hilbert-Schmidt, and thanks 
to~\cite[Theorem 6.23, pp. 210]{reed} they are always defined thanks to a kernel still denoted by $\rho \in L^2(\R^2)$, and such that 
\begin{equation}\label{op-kernel}
\rho :
\left\{ \begin{array}{rcl}
\ds  L^2(\R)& \ds \longrightarrow& \ds L^2(\R)\\
\ds \phi& \ds \longmapsto& \ds \rho (\psi)= \int_{\R} \rho(\cdot,X')\psi(X')dX'.
\end{array} \right.
\end{equation}
When $\rho$ is symmetric and nonnegative, the associated kernel is always defined almost everywhere on the diagonal $\{ X=X'\}$ and we have the following equality~\cite[Theorem 3.1]{brislawn}
\[
\Tr \rho = \int_\R \rho(X,X) \,dX. 
\]
However, we will not be able to prove all our results in the trace class $\SSS_1$, so will also need the Schatten classes $\SSS_p$  for $ p \in [1, +\infty]$, whose definitions are recalled before the statement of the mains results.

\medskip
\paragraph{ \bf A simple description for an interaction with the environment.}

We use an instantaneous description of interaction with ``light'' particles of the environment.
Such an instantaneous interaction will be modeled by the multiplication of the kernel $\rho$ by a suitable function:
\begin{align*}
\rho(X,X') \xrightarrow{collision} \rho(X,X') I_\chi^V(X,X')=: I_\chi^V[\rho](X,X').
\end{align*}
where the function $ I_\chi^V$ depends on the interaction potential $V$ and the wave function $\chi$ of the light particle. Actually that kind of model is rigorously obtained in~\cite{AHN} in a limit where the light particle is infinitely light (and fast).

Here we will consider only a non-physical case: when the reflection and transmission amplitudes naturally associated to the scattering operator are constant (independent of the wave number $k$): respectively $\alpha \in  [0,1]$ and $\pm i \sqrt{1 - \alpha^2}$ \footnote{In that particular case with constant amplitudes, considering $\alpha \in \C$ with $| \alpha | \le 1$ is useless since a phase factor $e^{i \theta}$ will not change anything in the calculation.}. Clearly this cannot happens for any potential $V$, since for instance the reflection amplitude should vanish in the limit $|k| \rightarrow + \infty$. But, despite this drawback, that assumption will not alter what is essential for our study. 
Moreover,  this is also an approximation of more realistic cases, which is quite good when the spread in velocity variable of the light particles is narrow. 

Under this assumption, and when the  light particle has a Gaussian wave function, with average position $x$, spread $\sigma$, and  average momentum $p$, the function $I^a_\chi$ take under that approximation the special form:
\begin{equation} \label{def:I}
I_{p,\sigma,x}^a(X,X')  = 1 -  \alpha^2 +  \alpha^2 \theta_{p,\sigma}(X-X') \pm i \alpha \sqrt{1 - \alpha^2}  e^{-2 \sigma^2 p^2} \bigl( \gamma_{\sigma}(X-x) -  \gamma_{\sigma}(X'-x)  \bigr),
\end{equation}
with
\[
\theta_{p,\sigma}(Y) :=  e^{2 i pY - \frac{Y^2}{2 \sigma^2}}  \quad\text{and}\quad \gamma_{\sigma}(X) :=  e^{ - \frac{X^2}{2 \sigma^2}}.
\]
This (super)-operator $I_{p,\sigma,x}^a$ preserves positivity, symmetry and  trace: 
\mbox{$\Tr \bigl(I_{p,\sigma,x}^a[\rho] \bigr) = \Tr \rho$}
(the last properties is a consequence of the fact than $I_{p,\sigma,x}^V(X,X)=1$ for all $X\in\R$). We refer to~\cite{AHN} for the details. More precisely, the map $\rho \mapsto I_{p,\sigma,x}^V[\rho]$ is completely positive, a property that is important if we want to construct a relevant quantum evolution. We refer to~\cite{AHN} for the details. 

\medskip
\paragraph{\bf A environment modeled by a thermal bath.}

Our heavy particle will encounter many random collisions with particles of the environment: at times  $T_i$ ($i\in\mathbb{N}^\ast$), it will interact  with a Gaussian particle with random spread, position, and momentum $(\sigma_i,x_i,p_i)$. A convenient way to model these interaction is to introduce a Poisson point process (PPP in the sequel) on $\R^+ \times \R^+ \times \R \times \R$: 
\begin{equation*} 
P_N = \sum_{j \in \mathbb{N}^\ast } \delta_{T_j,\sigma_j,x_j,p_j},
\quad \text{with intensity} \quad
 \frac{N}{2R}  \,  \indiq_{(0,+\infty)}(t) dt \otimes  \mu_s(d\sigma) \otimes \indiq_{[-R,R]}(x) dx \otimes \mu_m(dp),
\end{equation*}
where $N$ is the average number of interaction per unit of time, $R$ is a cut-off parameter unfortunately necessary in one dimension, and $\mu_s$ and $\mu_m$ are respectively the (normalized) distributions of spread and average momentum of the light particles. The distribution for $p$ and $\sigma$ can be general, but one very interesting case is when the light particles are assumed to be in a \emph{thermal bath} at temperature $\beta_0$.  
As explained for instance in~\cite[Section II.B]{HornSipe} a thermal bath at temperature $\beta_0$ is naturally obtained has  superposition of Gaussian wave packet with random (normal) momentum and fixed spread in the following way:
\begin{equation*} 
\mu_s= \delta_{\sigma_0}, \quad
\mu_m(dp) = \sqrt{\frac{\beta}{2\pi}}   \, e^{- \frac{\beta}2 p^2}\,dp,
\quad \text{with} \quad 
\frac1{\beta_0}= \frac1{\beta} + \frac1{4\sigma_0^2},
\end{equation*}
where $\delta$ stands for the Dirac distribution. Throughout the paper, for simplicity, we will assume that $\sigma_0=1$. Nevertheless, this assumption is not restrictive since a simple rescaling leads to the general result.  So to summarize, we assume that the parameter $(T_j, x_j, p_j)$   is a PPP with intensity measure
\begin{equation} \label{PPPbis}
\frac N{2R}   \,  \indiq_{(0,+\infty)}(t) dt  \otimes \indiq_{[-R,R]}(x) dx\otimes  \sqrt{\frac{\beta}{2\pi}}   \, e^{- \frac{\beta}2 p^2}\,dp,
\quad \text{with} \quad 
\frac1{\beta_0}= \frac1{\beta} + \frac14.
\end{equation}

\medskip
\paragraph{ \bf The appropriate scaling.}
We are interested by the limit when the average number of interaction per unit time goes to infinity, that is $N\to +\infty$. In order to get a non trivial limit, and in view of the form of $\theta_p$ in~\eqref{def:I}, it is natural to replace
$\alpha$ by $\alpha N^{-1/2}$.  In that case, \eqref{def:I}  becomes
\begin{equation} \label{approx:I}
I_{p,x}^N(X,X')  = 1 -  \frac{\alpha^2}N +  \frac{\alpha^2}N \theta_p(X-X') \pm i \frac{\alpha}{\sqrt N}\sqrt{1- \frac{\alpha^2}N} \, e^{-2 p^2} \bigl( \gamma(X-x) - \gamma(X'-x) \bigr),  
\end{equation}
with
\begin{equation}\label{def:ta} 
\theta_p(Y) :=  e^{2 i pY - \frac{Y^2}2}  \qquad \text{and} \qquad \gamma(X) :=  \,   e^{- \frac{X^2}2}.
\end{equation}

In terms of operators, we will also write for $\rho \in \SSS_1$
\begin{equation} \label{approx:Iop}
I_{p,x}^N [\rho]  = \Bigl(1 -  \frac{\alpha^2}N \Bigr) \rho +  \frac{\alpha^2}N \theta_p[\rho] \pm i \frac{\alpha}{\sqrt N}\sqrt{1- \frac{\alpha^2}N} \, e^{-2 p^2} \bigl[ \gamma(\cdot-x) \rho \bigr], 
\end{equation}
where $\gamma$ naturally act by multiplication and $\theta_p[\rho]$ is defined by
\begin{equation} \label{def:theta1}
\theta_p[\rho] (X,X') =\theta(X-X') \rho(X,X'), 
\quad \text{or} \quad  
\theta_p[\rho] := \frac1{\sqrt{2\pi}}  \int_\R e^{ik \cdot} \rho  e^{- ik \cdot}  \widehat \theta_p(k) \,dk, 
\end{equation} 
with the convention that $\widehat \theta_p(k) :=  \frac1{\sqrt{2\pi}}  \int_\R  e^{- ik \cdot} \theta_p(x) \,dx$.

\medskip
This is the good scaling, since the average number of interaction per unit time equals to $N$,  we can expect that the first term in $1/N$ will give a finite effect in the limit. Remark, that since $\theta_{p}$ does not depends on $x$, it is important here that we consider a PPP with a finite support in $x$. Otherwise, the number of interaction on any time interval will be infinite, and all the interactions will have a non vanishing effect and that will lead to a trivial limit (at least without any rescaling in position). For the last two terms involving $\gamma$, which scale in $1/\sqrt N$, we may also expect a finite but random limit. In fact, thanks to the translation invariance, these terms will have a vanishing expectation and we should therefore consider only their fluctuations which are finite. However, the previous argument is not fully correct because we have introduce a cut-off at size $R$ in the intensity measure of the  PPP, so the expectation of the two last term does not cancel exactly. To overcome that difficulty, which is essentially one dimensional, we have imagine several approximation and modification of our model. None of them was completely satisfactory but, an interesting one is probably the following: we compensate in our models the average value of the terms involving $\gamma$. We have no physical justification for the introduction of that artificial potential term: we introduce it only to obtain some result in the limit.
With respect to the others modifications we tried, it has the advantage to preserve the positivity of the density matrix $\rho$ which is a key property in quantum theory. 

Nevertheless, it is interesting to remark that this problem should not holds in larger dimension ($d \ge 2$), where the equivalent of the $\theta_p$ function does not depend on $X-X'$ any more, and has also a fast enough decrease at infinity. Hence, it should be allowed to use a full stationary (in position) environment in larger dimension. Of course, the collision process now becomes a jump process with infinitely many collisions during any finite time interval, but most of them have a very small effect. But this raised several difficulties and will be the aim of a future work. The present work and its ``unjustified'' approximations should be seen as a first step towards a more difficult but more relevant study in larger dimension.

\medskip
\paragraph{ \bf A Von-Neumann equation with jump.}

In the absence of interaction with the environment, the evolution of the heavy particle is described through the free Hamiltonian:
\begin{equation*} 
H_0 = - \frac12 \Delta, \qquad
D(H_0) = H^2(\R) = \bigl\{ \psi \in H=L^2(\R), \text{ s.t. }  \psi'' \in H \bigr\}.
\end{equation*}

For any $N \in \mathbb N$, Let us consider a \emph{PPP} $(T_j,x_j,p_j)_{j\in\mathbb{N}^\ast}$ as in~\eqref{PPPbis} and study the following stochastic Von-Neumann equation with jumps: for all $j\in\mathbb{N}$
\begin{equation} \label{VN-jump}
\left\{ \begin{array}{ll}
\ds  i \partial_t \rho^N = \Bigl[   H_0 \pm \alpha \sqrt {N- \alpha^2} \gamma_\infty, \rho^N \Bigr]  & \text{on } (T_{j},T_{j+1}) \\
\ds \rho^N_{T_j} =   I_{p_j,x_j} \bigl[ \rho^N_{T_j^-} \bigr]
\end{array} \right.,
\end{equation}
with $\rho^N_{T_0}=\rho^N_0=\rho_0$, and where $I^N_{p,x}$ is defined by~\eqref{approx:I}, and 
\begin{equation} \label{gamma_inf}
\gamma_\infty(X) := 
\int_\R \int_{-R} ^R e^{-2 p^2} \gamma (X-x)   \,\mu_m(dp)dp \,dx 
 = \frac{\sqrt{2 \pi \beta_0}}{4R} \Bigl( 
\erfc\Big(\frac{X-R}{\sqrt2}\Big) - \erfc\Big(\frac{X+R}{\sqrt2}\Big)\Bigr),
\end{equation}
with the function $\erfc$ defined by
$ \erfc(x) := \frac1{\sqrt \pi} \int_x^{+\infty} e^{-y^2}\,dy.$

The potential term appearing in~\eqref{VN-jump} seems very strong with its factor $\sqrt N$, but  it is  exactly compensated by the mean effects of collision, so that we obtain a finite limit despite this term (and rather thank to it). To observe this compensation, it is more convenient to rewrite that system with the help of the PPP: 
\begin{equation} \label{VN-PPP}
\rho^N_t = \rho^N_0 -i \int_0^t  \Bigl[   H_0 \pm \alpha \sqrt {N- \alpha^2} \gamma_\infty, \rho_s^N \Bigr]  \,ds
+ \intP \Bigl(  I^N_{p,x} \bigl[ \rho^N_{s^-} \bigr]  -\rho^N_{s^-} \Bigl)P_N(ds,dx,dp).
\end{equation}
Introducing the compensated Poisson measure 
\begin{equation}\label{compensated}
\tilde{P}_{N}(ds,dx,dp):=P_N(ds,dx,dp)- \frac N{2R} ds \otimes \indiq_{[-R,R]}(x)dx \otimes \mu_m(dp),
\end{equation}
the relation \eqref{VN-PPP} becomes
\begin{multline} \label{VN-PPPcomp}
\rho^N_t = \rho^N_0 -i \int_0^t  \bigl[   H_0 , \rho_s^N \bigr]  \,ds - \alpha^2 \int_0^t  \bigl( \rho^N_s -  \theta_\infty[\rho_s^N] \bigr) \,ds \\
+ \intP \Bigl(  I^N_{p,x} \bigl[ \rho^N_{s^-} \bigr]  -\rho^N_{s^-} \Bigl) \tilde P_N(ds,dx,dp),
\end{multline}
with $\theta_\infty$ defined by
\begin{equation} \label{theta_inf}
\theta_\infty(Y) :=  \int_\R \theta_p(Y) \mu_m(dp) = e^{-\frac2{\beta_0}Y^2}.
\end{equation}
In fact, using the definition~\eqref{approx:I}, we have
\[\begin{split}
\frac N{2R} \intPst \Bigl(  I^N_{p,x}  \bigl[ \rho^N_{s^-} \bigr]  -\rho^N_{s^-} \Bigl) \,\mu_m( dp)dx 
&= \frac {\alpha^2}{2R} \sqrt{\frac{\beta}{2\pi}} \intPst \Bigl(  \theta_p\bigl[\rho^N_{s^-}\bigr]  -\rho^N_{s^-} \Bigl) \,dx \,e^{-\beta  p^2/2} dp \\
& \pm i \frac {\alpha \sqrt{N-\alpha^2}}{2R} \sqrt{\frac{\beta}{2\pi}} 
\intPst \big[\gamma(\cdot -x) ,\rho^N_{s^-}  \big]  \,dx \,e^{- \bigl(2  +\frac12 \beta \bigr) p^2} dp  \\
&= \alpha^2 \Bigl(\theta_\infty\bigl[\rho^N_{s^-}\bigr] -  \rho^N_{s^-} \Bigr) 
\pm i  \alpha \sqrt{N-\alpha^2} \bigl[   \gamma_\infty , \rho^N_{s^-} \bigr],
\end{split}\]
by definition of $\theta_\infty$ \eqref{theta_inf} and of $\gamma_\infty$ \eqref{gamma_inf}.

\begin{rmk}
Here and below, we use the theory of integration in Banach space, developed originally by Bochner in~\cite{bochner}. It apply without difficulties since the trace class $\SSS_1$ is a separable space, which implies that any measurable function on a probability space with value in $\SSS_1$ is the pointwise limit of countably valued function. The same holds for $\DD([0,+\infty),\SSS_1)$ which is a separable Banach space under the Skorokod topology.

A important property of the Bochner integral is that it commutes with bounded linear operator. For instance for a random nonnegative symmetric operator $\rho$ with trace a.s. equal to $1$, and a bounded linear operator $T :\SSS_1 \mapsto B$ another Banach space, then
\[
T \bigl( \E[\rho] \bigr) = \E \bigl[ T(\rho) \bigr]. 
\]  
\end{rmk}

Now, let us remark that equations~\eqref{VN-jump}, \eqref{VN-PPP}, \eqref{VN-PPPcomp} are three different ways of writing the same evolution. In fact, \eqref{VN-jump} is obtained by integrating \eqref{VN-jump} in time, and \eqref{VN-PPPcomp} follows from~\eqref{VN-PPP} after a simple integration of the jump term again the intensity measure. The well-posedness for that tree equivalent systems is stated in the following result.
\begin{prop} \label{prop:WP-jump}
Let $N\in\mathbb{N}^\ast$, and assume that $\rho_0 \in \SSS_1^+$ and that $(T^N_j,x^N_j,p^N_j)_{j\in\mathbb{N}^\ast}$ is a PPP defined on some probability space with intensity measure given by~\eqref{PPPbis}. Then, for almost all realizations of  of the PPP, the $(T^N_j)_j$ are all distinct and without accumulation in finite time, and in that case there exists a unique solution in $L^\infty([0,+\infty),\SSS_1)$ to~\eqref{VN-jump}. This implies strong existence and uniqueness of an adapted process solutions to~\eqref{VN-jump}, or equivalently to~\eqref{VN-PPP} and~\eqref{VN-PPPcomp}.

Moreover, it holds almost surely that for any time $t\geq 0$, $\rho^N_t$ is a symmetric nonnegative operator with $\Tr \rho^N_t =1$.
\end{prop}

\begin{proof}
We denote by $e^{i t ( H_0 \pm \alpha \sqrt {N- \alpha^2} \gamma_\infty)}$ the group generated by $i ( H_0 \pm \alpha \sqrt {N- \alpha^2} \gamma_\infty)$ which is well defined since $\gamma_\infty$ is smooth and decrease very quickly to $0$ as $|X|$ goes to infinity~\cite{reed2}. Then, for a given realization of the PPP satisfying the hypothesis stated in the proposition, solutions of~\eqref{VN-jump} should satisfy between the jump (i.e. for $t \in [T^N_j,T^N_{j+1})$ for some $j$)
\[
\rho^N_t = e^{i (t-T^N_j)( H_0 \pm \alpha \sqrt {N- \alpha^2} \gamma_\infty) } \rho^N_{T^N_j} e^{- i (t- T^N_j) ( H_0 \pm \alpha \sqrt {N- \alpha^2} \gamma_\infty)}.
\]
Hence, the solutions of the evolution problem are uniquely characterized, since that formula defines the unique solution of the free Von Neumann equation that remains in $\SSS_1$ in-between two collisions. 

Remark also that the previous evolution between two collisions preserves the symmetry, the positivity, and the trace. Finally, the collision map $\rho \mapsto I^N_{p,x}[ \rho ]$ is also well defined from $\SSS_1$ into itself, and preserve symmetry, positivity and trace. We refer to Appendix \ref{Sch_Bd} for more details, and also to~\cite{AHN}.
\end{proof}

\medskip
\paragraph{\bf Main result}

Let us first recall that a compact operator on $L^2(\R)$ has always a norm convergent expansion
\begin{equation} \label{expansion}
\rho = \sum_{j=1}^{+\infty} \mu_j(\rho) | \psi_j \ra \la \phi_j |,
\end{equation}
where the $\mu_j(\rho) \ge 0$ stand for the singular values of $\rho$ listed in a decreasing order with multiplicities, and $(\phi_j)$ and $(\psi_j)$ are two (not necessarily complete) orthonormal sets \cite[Theorem 1.4]{simon}.
We borrowed the very convenient ``bra-ket'' notation to physicist. We recall the $| \psi \ra \la \phi |$ simply stands for the operator on $L^2: \chi \mapsto \la \phi, \chi \ra \psi$.

The values $\mu_j(\rho)$ are uniquely determined, and the two orthogonal families are up to some isometry on the eigenspaces of $\rho^\ast \rho$.  In the particular case where $\rho$ is symmetric, we also have $\psi_j=\pm \phi_j$ for each $j$ depending on the sign of the eigenvalues of $\rho$. Then the Schatten norm $\| \rho\|_{\SSS_p}$ of a compact operator is defined for any $p \in [1,+\infty]$ by
\begin{equation} \label{Sch_norm}
\|  \rho \|_{\SSS_p}:= \biggl( \sum_{j=1}^{+\infty} |\mu_j(\rho)|^p \biggr)^{\frac 1p},
\end{equation}
and the Schatten class $\SSS_p$ is the subset of compact operators for which the above series is convergent.

\medskip
Our main theorem state a result of convergence in the space of cad-lag  (right continuous with a limit from the left) process in the space $\DD\bigl([0,\infty), \SSS_p \bigr)$ for $p >1$, endowed with the Skorokhod topology~\cite{billingsley}.
\begin{thm}[Quenched convergence] \label{thm:main}
Let us assume that $\rho_0\in \SSS_1^+$ and
\begin{equation}\label{initcond}
\Tr \bigl( \inab \rho_0 \inab \bigr) + \Tr \bigl( X \rho_0 X) < + \infty,
\end{equation}
and consider $\rho^N$ the unique solution to~\eqref{VN-jump} for each $N$. Then, for $p \in (1,2]$, it converges in law, as $N \rightarrow +\infty$, in $\DD([0,\infty),\SSS_p)$ towards the unique  (weak) solution  of the stochastic ``Lindblad'' equation
\begin{equation} \label{LB-sto}
i \,d \rho_\infty = \bigl[   H_0\, dt + dW_t, \rho_\infty \bigr]
+ i\, \alpha^2 \Bigl( \theta_\infty[\rho_\infty]- \rho_\infty  \Bigr  )\ \,dt,
\end{equation}
where $\theta_\infty$ act by multiplication on the kernel as $\theta_\infty[\rho_\infty](X,X') =  \rho_\infty(X,X') e^{-2\frac{(X-X')^2}{\beta_0}}$, or equivalently 
\[
\theta_\infty[\rho_\infty]  = \int_{\R}e^{ip \cdot } \rho_\infty e^{-ip\cdot } \, \mu_m(dp),
\]
and $W$ is a Brownian potential with correlation function given by
\begin{equation} \label{eq:cor}
\E \bigl[  W_s(X) \, W_t(X') \bigr] := (s \wedge t)  \frac{\alpha^2  \sqrt \pi}{2R} \sqrt{ \frac{\beta_0}{8 -  \beta_0}}
e^{-\frac12 (X-X')^2}  
\biggl( \erfc\Bigl( -R - \frac{X+X'}2\Bigr)  - \erfc\Bigl( R - \frac{X+X'}2\Bigr)   \biggr).
\end{equation}
\end{thm}

\begin{rmk}
The dissipative term as the Lindblad form~\cite{Lindblad}: It may precisely be rewritten
\[
\theta_\infty[\rho_\infty] - \rho_\infty = \int_{\R} \bigl( e^{ip \cdot } \rho_\infty e^{-ip\cdot } -1 \bigr) \, \mu_m(dp),
\]
so it is a linear combination of terms of the form $U \rho_\infty U^* - \rho_\infty$, with isometric $U$. 
\end{rmk}

\begin{rmk} An intuitive derivation of the Brownian potential is that $N^{-1/2} \tilde P_N$ converges in law towards a Gaussian white noise $\dot B$ in time, position and velocity, i.e.\ with correlation function
\[\E[\dot B_s(x,p) \dot B_t(y,q)]= \delta(s-t)  \frac{1}{2R}\indiq_{[-R,R]}(x) \delta(x-y) \mu_m(p)\delta(p-q).\]
Then, the stochastic potential $dW$ may be written as
\begin{equation} \label{GField}
dW_t(X) = \pm \alpha  \int_{\R^2}   e^{-2 p^2} \gamma(X-x) \dot W(t,x,p) \, dp \,dx,
\end{equation}
and then \eqref{eq:cor} comes from the following computation
\begin{align*}
\E \bigl[  W_s(X) \, W_t(X') \bigr]  & = \alpha^2 \int_0^s\int_0^t \E[ dW_u(X) dW_v(X') ]dudv\\
&=\alpha^2  \int_0^s\int_0^t \int_{\R^4} e^{-2 (p^2+q^2)} \gamma(X-x) \gamma(X'-y) \E \bigl[\dot B(u,x,p) \dot B(v,y,q) \bigr] \,dx\,dy\,dp\,dq \\
& =(s \wedge t) \frac{\alpha^2}{2R}  \sqrt{\frac{\beta}{2\pi}} \intPst    e^{-\frac12 \bigl[ (X-x)^2+ (X'-x)^2 \bigr]} e^{- \bigl( 4+ \frac 12 \beta \bigr) p^2} dx\, dp \\
& =  (s \wedge t)  \frac{\alpha^2}{2R} \sqrt{ \frac{\beta_0}{8 -  \beta_0}} e^{-\frac12 (X-X')^2}
 \int_{-R}^R  e^{- \bigl( x - \frac{X+X'}2\bigr)^2}   \,dx.
\end{align*}
\end{rmk}

The term $ i \alpha^2 \bigl( \rho_\infty - \theta[\rho_\infty] \bigr)$ is the decoherent term: it roughly decrease the ``off-diagonal'' terms of the density operator. 
But since the random potential term act in the It\=o sense, the evolution equation without the decoherent term is non reversible. In order to transform it in a reversible evolution plus some dissipative term, we have to switch to the the Stratonovich representation: 
\begin{multline*} 
i d \rho_\infty = \bigl[   H_0 dt + \circ \, dW_t , \rho_\infty \bigr]
+ i\, \alpha^2 \bigl(  \theta_\infty[\rho_\infty] - \rho_\infty \bigr) \, dt \\
+ i \, \frac{\alpha^2}2 \sqrt{\frac{\beta_0}{8 - \beta_0}}\, \int_{-R}^R \bigl[  \gamma( \cdot -x), 
\bigl[  \gamma (\cdot -x) , \rho_\infty \bigr]
 \bigr] \,dx \,dt.
\end{multline*}
Remark that this is again a Lindblad super-operator but with the bad sign: the double commutator is equal to $[\gamma, [\gamma, \rho_\infty]] = - (2 \gamma \rho \gamma - \gamma^2 \rho - \rho \gamma^2 )$. 
The new term may also be written $-\frac12 \E\bigl[ [dW_t [ dW_t, \rho_\infty]] \bigr]$, and has an explicit but complicated formula in terms of the $\erfc$ function. This complicated form is due to the cut-off, and does not give much information.

However, in kernel formulation it corresponds to the  multiplication by (up to a factor) :
\begin{equation}\label{def:gamma_inf}
\gamma_\infty(X,X') :=  \frac12 \sqrt{\frac{\beta_0}{8 - \beta_0}}\, 
\biggl(\int_{-R}^R \bigl(\gamma( X -x) -  \gamma (X' -x) \bigr)^2  
\,dx \biggr).
\end{equation}
Denoting by $\gamma_\infty[\rho_\infty]$ the operator with kernel $\gamma_\infty(X,X') \rho_\infty(X,X')$, we end up with the above Stratonovich version of~\eqref{LB-sto}:
\begin{equation} \label{LB-strat}
i d \rho_\infty = \bigl[   H_0 dt + \circ \, dW_t , \rho_\infty \bigr]
+ i\, \alpha^2 \bigl(\theta_\infty[\rho_\infty] + \gamma_\infty[\rho_\infty]  -  \rho_\infty \bigr) \, dt
\end{equation}
Since the new function $\gamma_\infty$ defined in~\eqref{def:gamma_inf} vanishes on the diagonal $\{ X=X'\}$ and is positive outside, the associated in the above equation is re-coherent. 

The Stratonovich formulation~\eqref{LB-strat} separate the dynamics 
in a reversible evolution (Hamiltonian with noise in the Stratonovich sense) and a dissipative term. So the ``true'' decoherent term is 
actually the one appearing in~\eqref{LB-strat} and not the one of~\eqref {LB-sto}: for instance it is that term that should be taken into account to calculate the decrease of the ``off-diagonal'' terms of the density operators.

\medskip
\paragraph{\bf Annealed convergence}

If we are interested only in the mean behaviour of $\rho^N$, then we can take the expectation in~\eqref{VN-PPPcomp}, and obtain, since the integral with respect to the compensated PPP is a martingale, that 
\begin{equation} \label{eq:ann}
\E [\rho^N_t] = \rho^N_0 -i \int_0^t  \bigl[   H_0 , \E[\rho_s^N] \bigr]  \,ds -  \alpha^2 \int_0^t  \Bigl( \E[\rho^N_s] -  \theta_\infty\bigl[ \E[\rho_s^N] \bigr] \Bigr) \,ds.
\end{equation} 
To rigorously commute the expectation with $H_0$, we will introduce the unitary group generated by $i H_0$, which does commute with the expectation as a continuous linear application (see below). This  equation is deterministic and independent of $N$, so it is certainly satisfied also by the limit process. In fact, taking the expectation in the  It\=o formulation~\eqref{LB-sto} of the limit equation also leads to~\eqref{eq:ann}: it is enough to erase the stochastic terms. So the annealed convergence is formally obvious. To do it rigorously, it is enough to prove that the (deterministic) annealed equation~\eqref{eq:ann} admits a unique solution. But, once the free evolution is filtrated, it becomes a linear evolution equation, involving only a bounded (super)-operator: its solutions are clearly given by the exponential of this operator. We will not give more details on that point, and concentrate in the rest of the paper on the more difficult problems raised by quenched evolution.

Remark that the annealed equation is more decoherent that the quenched one, since the (re-coherent) $\gamma_\infty$ term of~\eqref{LB-strat} has disappeared in~\eqref{eq:ann}. In fact in the annealed model, we do not care about the correlations due to the ``common history'' shared by distant points; in the sense that  for any realization of the noise, the density operator defined in the whole space evolves according to its global equation.

\medskip
\paragraph{\bf Existing literature and interest of that new model.}

The main interest of this model is twofold. First, it is to the best of our knowledge the first rigorous derivation of a Lindblad equation. In fact, most of the works on the decoherence induced by the environment are found in the physic literature~\cite{JoosZeh,Diosi,AMS,DoddHall,HornSipe,Adler}. But of course, we do it from a toy model, and the derivation from a more realistic dynamics seems very challenging.

Moreover we naturally end up with a model which incorporate time fluctuations via a random time dependent ``effective'' potential. In fact, in the previous physicist literature~\cite{JoosZeh,Diosi,AMS,DoddHall,HornSipe,Adler}, the computation of exact decoherence coefficients was performed without time fluctuations: these authors do not use any PPP (or something similar) for the description of the environment, but rather use only the intensity measure~\eqref{PPPbis}, i.e. only the mean effect. In short, they derive some annealed models like~\eqref{eq:ann}.
So we show here that taking into account the time fluctuations of the environment, randomness remains in the limit, but in a non obvious way, i.e. as a random ``effective'' potential. 
Moreover and because of the random term, the effective decoherence in the quenched model is lower than the decoherence given by the annealed equation. 

\medskip
\paragraph{\bf Plan of the paper}
In Section~\ref{section:tech}, we collect some preliminary results about the effect of one collision, and the evolution of the kinetic energy and  the momentum in position. The Section~\ref{proof1} is devoted to the proof of Theorem~\ref{thm:main}. Finally, several appendix give some important lemmas and properties: on the norm of important operators in the Schatten class, on martingales valued in $\SSS_2$, on compact subsets of $\SSS_1$.

%
%
%
%
\section{Preliminary results}\label{section:tech}

Let us start by proving that it is possible to exchange unbounded linear operator with the Bochner integration (written here as expectation). 
\begin{lem} \label{lem:exchange}
\begin{enumerate}
\item[(i)] Assuming that $\rho$ is a random symmetric nonnegative operator (that is it belongs to $\SSS_1^+$), we have the two following equalities:
\[
\E\Bigl[ \Tr \bigl( \inab \rho \inab \bigr) \Bigr] =  
\Tr \bigl( (i\nabla) \E[\rho]  (i\nabla) \bigr),
 \qquad 
  \E\bigl[   \Tr \bigl( X \rho X) \bigr] =  
  \Tr \bigl( X \E[\rho]  X \bigr).
\]
\item[(ii)] Let $(\rho_n)_{n \in \N}$ be a sequence of random symmetric nonnegative operator converging weakly in law towards some $\rho$, (weakly means that for any finite rank operator $A$, $\la \rho_n, A\ra_{\SSS_2} \rightarrow \la \rho, A\ra_{\SSS_2}$  in law). Then, we have
\[
\E\Bigl[ \Tr \bigl( \inab \rho \inab \bigr) \Bigr] \le 
\liminf_{n \to \infty} \E\Bigl[ \Tr \bigl( \inab \rho_n \inab \bigr) \Bigr],
\qquad
\E\bigl[   \Tr \bigl( X \rho X) \bigr] \le 
\liminf_{n \to \infty}
\E\bigl[   \Tr \bigl( X \rho_n X) \bigr].
\]
\end{enumerate}
\end{lem}
Remark that the usual property of the Bochner integral may not be applied since the linear functionals $\rho \mapsto \Tr \bigl( \inab \rho \inab \bigr)$ and $\rho \mapsto \Tr \bigl( X \rho X)$ are not bounded linear functionals. But, when $\rho$ is a symmetric nonnegative operator, they are at least defined in $[0,+\infty]$, and we can improve the result as in the cases we integrate nonnegative functions on $\R$ rather than general functions without specific sign.

\begin{proof}[of Lemma \ref{lem:exchange}]
Here, we restrict ourselves to the proof of the results regarding the kinetic energy. The proofs to obtain the corresponding result for the momentum in position follow the same lines. 

\emph{Item (i).} First, Let us pick up a basis $(\psi_n)_{n \in \N}$ of $L^2(\R)$ made of smooth functions: use for instance the smooth wavelets introduced by Lemari\'e and Meyer~\cite{lemarie}. Then, using that $i \nabla$ is self-adjoint, we may write
\[
\Tr \bigl( \inab \rho \inab \bigr) =
\sum_{j = 1}^{+\infty} \bigl\la \psi_j \big| \inab \rho \inab \big| \psi_j \bigr\ra
=
\sum_{j = 1}^{+\infty} \bigl\la \inab \psi_j \big|  \rho  \big| \inab \psi_j \bigr\ra,
\]
where the equalities are still valid even if the sum in the r.h.s. is infinite since all the terms are non negative. Introducing the following notation for the partial sum
\begin{equation} \label{lambda_M}
\lambda_M(\rho ) :=  \sum_{j=1}^M \bigl\la \inab \psi_j \big|  \rho  \big| \inab \psi_j \bigr\ra,
\end{equation}
it is clear that for any $\rho \in \SSS_1^+$, 
\[\lim_{M\to+\infty}\lambda_M(\rho)= \Tr \bigl( \inab \rho \inab \bigr),\]
and for each $M$, $\lambda_M$ is also a bounded linear mapping form $\SSS_1$ to $\R$, so that thanks to the usual properties of the Bochner integral we have $\E\bigl[\lambda_M(\rho)\bigr] = \lambda_M \bigl( \E[\rho]\bigr)$. Now, using that $\E[\rho]$ is still symmetric and nonnegative, we also have 
 \[\lim_{M \rightarrow +\infty} \lambda_M \bigl(\E[\rho]\bigr) = \Tr \bigl( \inab \E[\rho] \inab \bigr),\qquad 
\text{and}\qquad\lim_{M\to+\infty}\E\bigl[\lambda_M(\rho)\bigr]=\E\bigl[ \Tr \bigl( \inab \rho \inab \bigr) \bigr],\]
by Lebesgue's monotone convergence theorem.

\emph{Item (ii).} According to \cite[ Theorem 3.4 pp. 31]{billingsley}, the convergence in law of $\lambda_M(\rho_n)$ towards $\lambda_M(\rho)$ for any $ M \in \N$, implies that
\[
\liminf_{n \to \infty} \E \bigl[ \Tr \bigl( \inab \rho_n \inab \bigr) \bigr] \ge 
\liminf_{n \to \infty} \E \bigl[\lambda_M(\rho_n) \bigr] \ge \E \bigl[\lambda_M(\rho) \bigr].
\]
Then, from item $(i)$, $\lim_{M \to \infty} \E \bigl[\lambda_M(\rho) \bigr] = \E \bigl[ \Tr \bigl( \inab \rho \inab \bigr) \bigr]$, and the conclusion follows.
\end{proof}

\subsection{The effect of one  collision} \label{subsec:1col}

This section states some preliminary results which are needed in the proof of Theorem \ref{thm:main} to obtain a tightness result for $(\rho_N)_N$. The first one deals with the effect of only one collision on the density operator.

\begin{lem} \label{lem:I_bound}
For all $q \ge 1$, $p\in\R$, $x\in [-R,R]$, and $\rho\in\SSS_q$, we have for $I^N_{p,x}$ is defined by~\eqref{approx:I}:
\begin{equation} \label{eq:L2col}
 \bigl\|  ( I^N_{p,x} -Id)[\rho] \bigr\|_{\SSS_q} 
 \le 2\,   \Bigl(  \frac{\alpha^2}N + \frac{\alpha}{\sqrt N} e^{-2 p^2} \Bigr)  \, \|  \rho \|_{\SSS_q},
\end{equation}
If moreover, $\Tr \bigl( \inab \rho \inab \bigr) < + \infty$, then 
\begin{equation}
\Bigl| \Tr \Bigl( \inab I^N_{p,x}[\rho] \inab \Bigr)  - 
\Tr \Bigl( \inab \rho \inab \Bigr) \Bigr| \le 
 \frac {\alpha + \alpha^2}{\sqrt N} (1+8p) \sqrt{\Tr \bigl( \inab \rho \inab \bigr)} + 
  \frac{\alpha^2}N \Bigl( 1 + 2p^2 \Bigr).
\end{equation}

\end{lem}

\begin{proof}

 The first point is a simple application of $(iii)$ and $(iv)$ of Proposition~\ref{prop:Schatten} given in Appendix~\ref{Sch_Bd}. In fact, we have by definition
\[
( I^N_{p,x} -Id)[\rho]  = - \frac{\alpha^2}N \rho + \frac{\alpha^2}N \theta_p[\rho] \pm i e^{-2p^2} \frac{\alpha}{\sqrt N} \sqrt{1 - \frac{\alpha^2}N} \bigl( \gamma \rho - \rho  \gamma \bigr), \]
and then,
\[\begin{split}
\|( I^N_{p,x} -Id)[\rho]\|_{\SSS_q} &\le  \frac{\alpha^2}N \| \rho \|_{\SSS_q}+ \frac{\alpha^2}N \bigl\|\theta_p[\rho] \bigl\|_{\SSS_q} +  \frac{\alpha \, e^{-2p^2}}{\sqrt N} \bigl( \|  \gamma \rho \|_{\SSS_q} + \| \rho  \gamma \|_{\SSS_q} \bigr) \\
& \le \frac{\alpha^2 }N   
 \biggl(1+ \frac{\bigl\| \widehat \theta _p\bigr\|_1}{\sqrt{2 \pi}}\biggr)\| \rho \|_{\SSS_q}
+ \frac{2 \alpha\, e^{-2p^2} }{\sqrt N} \| \gamma \|_\infty \| \rho  \|_{\SSS_q}.
\end{split}\]
The conclusion follows using the definitions~\eqref{def:ta}: $\bigl\| \widehat \theta_p \bigr\|_1 = \sqrt{2 \pi}$
and $\| \gamma \|_\infty = 1$.

\medskip
The inequality on the kinetic energy is proved as follows. First, from~\eqref{approx:I}
\begin{multline*}
\Tr \bigl( \inab I^N_{p,x} [\rho] \inab  \bigr)- \Tr \bigl( \inab \rho \inab \bigr) \Bigr) =  \frac{\alpha^2}N  \Bigl(\Tr \bigl( \inab \theta_p [\rho] \inab \bigr) - \Tr \bigl( \inab \rho \inab \bigr) \Bigr) \\
\pm i e^{-2p^2}\frac\alpha{\sqrt N} \sqrt{1 - \frac{\alpha^2}N}  \Tr \Bigl( \inab  [ \gamma,\rho] \inab  \Bigr)
\end{multline*}
Next using point $(iii)$ of Lemma~\ref{kin_eq} with $\theta_p(Y) = e^{2i pY - \frac{Y^2}2}$, which satisfies $\theta_p(0) =0$, $\theta_p'(0) = 2i p$, $\theta_p''(0) = -(1 + 4 p^2)$, we get
\[
\Tr \bigl( \inab \theta_p [\rho] \inab \bigr) - \Tr \bigl( \inab \rho \inab \bigr) 
=- 4p \Tr \bigl( \inab \rho\bigr) + (1+ 4p^2).
\]
Then Lemma~\ref{lem:mix} implies that $\bigl| \Tr \bigl( \inab \rho\bigr) \bigr| \le 2 \sqrt{\Tr \bigl( \inab \rho \inab \bigr)}$. And finally
\begin{equation} \label{eqcin_step1}
\Bigl|\Tr \bigl( \inab \theta_p [\rho] \inab \bigr) - \Tr \bigl( \inab \rho \inab \bigr)\Bigr| \le   8 p \sqrt{ \Tr \bigl( \inab \rho \inab \bigr) }+ (1 + 4p^2).
\end{equation}
To treat the remaining term, 
we use $\bigl[\inab, \gamma \bigr] = i \gamma'$ and write
\[
\inab [\gamma, \rho] \inab = \bigl[ \gamma, \inab \rho \inab \bigr] +  i \bigl( \gamma' \rho \inab + \inab \rho \gamma'\bigr)
\]
The first term does belongs to $\SSS_1$, thanks to the assumption $Tr \bigl( \inab \rho \inab \bigr) < + \infty$ and point $(iv)$ of Proposition~\ref{prop:Schatten}.
And moreover, its trace vanishes, since it is a bracket.
The second term is bounded with the help of Lemma~\ref{lem:mix}, 
\begin{align*}
\Tr \bigl( \gamma' \rho \inab + \inab \rho \gamma' \bigr) & \le
2 \sqrt{\Tr\bigl( \gamma' \rho \gamma' \bigr) \Tr\bigl( \inab \rho \inab \bigr)} \\
& \le 2 \sqrt{ \| \gamma'\|_{\SSS_\infty}^2\Tr( \rho\bigr) \Tr\bigl( \inab \rho \inab \bigr)}
\le 2 \| \gamma'\|_\infty \sqrt{ \Tr\bigl( \inab \rho \inab \bigr)}.
\end{align*}
All in all, since $\| \gamma'\|_\infty \le 1$, we end up with
\[
\Bigl| \Tr \Bigl( \inab I^N_{p,x}[\rho] \inab \Bigr)  - 
\Tr \Bigl( \inab \rho \inab \Bigr) \Bigr| \le 
\frac {\alpha + \alpha^2}{\sqrt N} (1+8p) \sqrt{\Tr \bigl( \inab \rho \inab \bigr)} + 
  \frac{\alpha^2}N \Bigl( 1 + 2p^2 \Bigr),
\]
which is the claimed inequality.
\end{proof}

\subsection{The evolution of kinetic energy and momentum in position}

The following proposition describes the evolution of the expected values for the kinetic energy and the momentum in position under the effect of collisions. We will rely on Lemma~\ref{lem:exchange} that allow to exchange expectation and unbounded operators under some positivity assumption: the mean kinetic energy and the mean momentum in position are given respectively by the kinetic energy and momentum in position of the mean density operator.  

\begin{prop}\label{propSPDE}
Assume that $(\rho^N_t)_{t \ge 0}$ is a solution of~\eqref{VN-jump} satisfying~\eqref{initcond}. We have for $t \ge 0$ 
\begin{align}\label{unifbound1}
  \E\Big[ \Tr \bigl( (i\nabla) \rho^N_t (i\nabla)  \bigr) \Big] &=  \Tr \bigl( \inab \rho_0 \inab \bigr) + 
  \frac{4 \alpha^2}{\beta_0} t, \\
 \label{unifbound2} 
  \E\Big[   \Tr \bigl( X \rho^N_t X) \Big]  & =  \Tr \bigl( X \rho_0 X) + 
  \Tr\bigl( \inab \rho_0 X + X \rho_0 \inab \bigr) \, t +
  \Tr \bigl( \inab \rho_0 \inab \bigr) \, t^2 +
  \frac{4 \alpha^2}{3 \beta_0} t^3,
\end{align}
  where
  \[ \big\vert \Tr\bigl( \inab \rho_0 X + X \rho_0 \inab \bigr)\big\vert  \leq2 \sqrt{ \Tr \bigl( X \rho_0 X) \, \Tr \bigl( \inab \rho_0 \inab \bigr)}<+\infty. \]
 \end{prop}
 
 \begin{proof}

Using Lemma~\ref{lem:exchange}, we can use the deterministic evolution equation~\eqref{eq:ann} on $\rho^{av}_t := \E\bigl[   \rho^N_t  \bigr]$. We first multiply it by $\inab$ on the right and on the left, and take the trace. Since the free unitary group preserves the kinetic energy, we then obtain 
\[
 \partial_t  \Tr \bigl( \inab \rho^{av}_t \inab \bigr)  = - \alpha^2
\Tr \bigl( \inab \rho^{av}_t \inab \bigr) 
+ \alpha^2 \Tr \bigl( \inab \theta_\infty \bigl[ \rho^{av}_t \bigr] \inab \bigr) .
\]
Next, we want to use $(ii)$ of Lemma~\ref{kin_eq} with $\theta_\infty$, which satisfies $\theta_\infty(0) =1$, $\theta_\infty'(0)=0$ and $\theta_\infty'' (0)=  - 4 \beta_0^{-1}$. Hence, using also that $\Tr \rho^{av}_t =1$, we obtain
\begin{equation}\label{kin_ep_theta}
\Tr \bigl( \inab \theta_\infty \bigl[ \rho^{av}_t \bigr] \inab \bigr) = 
\Tr \bigl( \inab \rho^{av}_t \inab \bigr) +
\frac 4{\beta_0},
\end{equation}
and the first equality follows easily.

To prove the second one, we first multiply by $X$ on the right and on the left of~\eqref{eq:ann} and take the trace to obtain
\[
 \partial_t  \Tr \bigl( X \rho^{av}_t X \bigr)  = 
 \Tr \Bigl(X \bigl[ H_0, \rho^{av}_t \bigr] X  \Bigr)
 - \alpha^2
\Tr \bigl( X \rho^{av}_t X \bigr) 
+ \alpha^2 \Tr \bigl( X \theta_\infty \bigl[ \rho^{av}_t \bigr] X \bigr).
\]
Now, using $(i)$ of Lemma~\ref{kin_eq} in Appendix~\ref{Sch_Bd}, again applied to $\theta_\infty$, the second term in the r.h.s. vanishes so that
\[
 \partial_t  \Tr \bigl( X \rho^{av}_t X \bigr)  = 
 \Tr \Bigl(X \bigl[ H_0, \rho^{av}_t \bigr] X  \Bigr),
 \]
which means in short that the quantity $\Tr \bigl( X \rho^{av}_t X \bigr) $ evolves as if 
$\rho^{av}_t$ is not interacting with the environment, and then evolves freely. Then, following the proof of the first equality of Lemma~\ref{free_case} in Appendix \ref{Appendix:B} we have
\[
\partial_t \Tr ( X \rho^{av}_t X) = \Tr\bigl( \inab \rho^{av}_t X + X \rho^{av}_t \inab \bigr).
\]
Moreover, following the proof of the second point of Lemma~\ref{free_case} together with $(iii)$ of Lemma~\ref{kin_eq}, we obtain
\[
\partial_t \Tr\bigl( \inab \rho^{av}_t X + X \rho^{av}_t \inab   \bigr) = 2 \Tr \bigl( \inab \rho^{av}_t \inab \bigr),
\]
since  $\theta_\infty(0) =1$, $\theta_\infty'(0)=0$.

To conclude, it only remains to prove that $\Tr\bigl( \inab \rho_0 X + X \rho_0 \inab \bigr) < +\infty$ when 
$\Tr \bigl( X \rho_0 X) +   \Tr \bigl( \inab \rho_0 \inab \bigr) < + \infty$. This is a consequence of lemma~\ref{lem:mix} in Appendix~\ref{Sch_Bd} applied with $g(X)=X$:
\[
\Bigl[\Tr\bigl( \inab \rho_0 X + X \rho_0 \inab \bigr) \Bigr]^2 \le 
4 \Tr \bigl( X \rho_0 X) \, \Tr \bigl( \inab \rho_0 \inab \bigr).
\] 
Let us remark that there is another way to obtain that last bound. For the free evolution, \eqref{unifbound2} must be satisfied without the term involving $\alpha$ and $\beta$, so that the term in the r.h.s. should remains nonnegative for any time $t$, which leads directly to the above inequality.
\end{proof}

%
%
%
%
\section{Proof of Theorem \ref{thm:main}} \label{proof1}

The proof of Theorem \ref{thm:main} is based on martingale techniques and is in two steps. First, we prove the tightness of $(\rho^N)_N$, and then we characterize all the subsequence limits thanks to a well-posed stochastic differential equation on $\SSS_2$. Nevertheless, we do not directly work with $(\rho^N)_N$. We introduce an auxiliary process which is more convenient to use standard results on tightness and stochastic differential equation on Hilbert space.

\medskip
\paragraph{\bf  Filtration of the free transport.} W defined the following filtered process 
\begin{equation} \label{PJ_proc}
\trho^N_t:= T_{-t} \bigl[ \rho^N_t \bigr] := e^{i t H_0 } \rho^N_t e^{-i t H_0},
\end{equation}
where $\rho^N$ satisfies \eqref{VN-PPPcomp}, and $T_t$ is an isometry on  $\SSS_p$ for all $p \ge 1$ by Proposition~\ref{prop:Schatten} in Appendix \ref{Sch_Bd}. The interest is that the process $\trho^N$ evolves according to the following dynamics:
\begin{multline} \label{SPDEpoisson}
 \trho^N_{t}  = \rho_0 - \alpha^2 \int_0^{t}  
 \Bigl(  \trho^N_s  - T_{-s} \theta_\infty\bigl[ T_s [\trho^N_s] \bigr]  \Bigr)\, ds\\
 +  \intP \Bigl( T_{-s} I_{p,x}\bigl[ T_s[ \trho^N_{s^-}]\bigr]-\trho^N_{s^-} \Bigr) \, \tilde P_N(ds,dx,dp),
 \end{multline}
where no unbounded operator appears. 
In the two following sections, we work with $\trho^N$ in a first time, and then we derive the desired properties for $(\rho^N)_N$. 

\subsection{Tightness}\label{tightnesssec}
We begin this section by proving a uniform estimate for the kinetic energy. 

\begin{prop}\label{unifkinetic} Let $(\rho^N_t)_{t \ge 0}$ be the unique process solution to~\eqref{VN-jump}. 
 For any $T>0$,
\[
\sup_N \E\Big[\sup_{t\in[0,T]} \Tr(\inab \rho^N_t \inab)^2 \Big]< + \infty.
\]
\end{prop}
 This estimate is needed to obtain a tightness result for $(\rho^N)_N$ at the end of this section. 

\begin{proof} 
From the Definition~\eqref{PJ_proc} of the filtered process, it comes $\Tr \bigl( \inab \trho^N_t \inab \bigr) = \Tr \bigl( \inab \rho^N_t \inab \bigr)$. 
Using this, the equations~\eqref{SPDEpoisson} and~\eqref{kin_ep_theta}, we get that
\begin{multline} \label{kin_without_sup}
Tr\bigl(\inab \rho^N_t \inab \bigr) = Tr \bigl(\inab \rho^N_0 \inab \bigr) + \frac{4 \alpha^2 t}{\beta_0} \\
+  \intP \Tr \Bigl[ \inab \Bigl( I_{p,x}\bigl[ \rho^N_{s^-} \bigr]-\rho^N_{s^-} \Bigr)  \inab \Bigr] \, \tilde P_N(ds,dx,dp).
\end{multline}
Let us denote by $Q^N_t$ the stochastic integral of the r.h.s.\ . By Doob inequality and the second point of Lemma~\ref{lem:I_bound}, we have
\begin{align*}
\E\biggl[\sup_{t\in[0,T]}\Bigl \vert \intP  \Tr \Bigl[ & \inab \Bigl( I_{p,x}\bigl[ \rho^N_{s^-}  \bigr]-\rho^N_{s^-} \Bigr)  \inab \Bigr] \tilde{P}_N(ds,dx,dp) \Bigr \vert^2 \biggr] \\
& \le 2\,  \E\biggl[ \Bigl \vert \int_0^t \intP  \Tr \Bigl[ \inab \Bigl( I_{p,x}\bigl[ \rho^N_{s^-} \bigr]-\rho^N_{s^-} \Bigr)  \inab \Bigr] \tilde{P}_N(ds,dx,dp) \Bigr \vert^2 \biggr], \\
& = 2  \frac N {2R} \int_0^t   \intPst   \E \Bigl[ \Tr \Bigl[ \inab \Bigl( I_{p,x}\bigl[ \rho^N_{s^-} \bigr]-\rho^N_{s^-} \Bigr)  \inab \Bigr]^2 \Bigr] \, dx\,\mu_m(dp)   \,ds \\
& \le \frac C {2R}  \int _0^t  \E \Bigl[  \intPst \Tr\bigl(\inab \rho^N_s \inab \bigr) ^2 \, dx\, (1+p^4)\mu_m(dp)  \Bigr] \,ds, \\ 
& \le C \int_0^t  \E \Bigl[ \Tr\bigl(\inab \rho^N_s \inab \bigr) ^2 \Bigr] \,ds.
\end{align*}
Taking the square and the supremum in time in~\eqref{kin_without_sup}, and using the above bound, we get
\[
\E\Bigl[\sup_{t\in[0,T]} \Tr\bigl(\inab \rho^N_t \inab \bigr)^2 \Bigr] \le C \biggl( 1  + T^2 +  \int_0^T \E\Bigl[\sup_{s\in[0,t]} \Tr\bigl(\inab \rho^N_s \inab \bigr)^2 \Bigr]\,ds \biggr),
\]
and an application of the Gr\"onwall lemma conclude the proof. 
Actually, to proof this, we should first check that the expectation $\E \Bigl[ \Tr\bigl(\inab \rho^N_s \inab \bigr) ^2 \Bigr]$ is not infinite when we integrate on time intervals, so that our bounds have a meaning. But this can be done using similar and simpler arguments.
\end{proof}

We use here a standard tightness criterion which is valid for process with values in general separable Banach spaces, and particularly for the Schatten class $\SSS_p$ with $p \in(1,2]$.  The proposition below  is stated for instance in~\cite[Theorem 3 pp. 47]{kushner}, and goes back to Kurtz~\cite{kurtz} and Aldous~\cite{aldous1}.
\begin{prop}\label{aldous} 
Let $(Z_N)_{N \in \N}$ be a sequence of random variables in $\DD([0,\infty),B)$, for a separable Banach space $B$, endowed with a norm $\|\cdot\|_B$. For any $t \ge 0$, we denote by $\E^N_t$ the conditional expectation with respect to the filtration $\mathcal{F}^N_t=\sigma(Z_N(s),\, s\leq t)$.
If 
\begin{itemize}
\item[(i)] at any time $t\in \R^+$, $Z^N(t)$ is a tight sequence in $B$;
\item[(ii)] for any $N$, $T \in \R^+$, and $h\in(0,1)$, there exists a real-valued random variable $\eta^T_N(h)$ such that
\[ \forall \; t \in [0,T], \quad 
\sup_{0 \le u \le h }\, \E^N_t \bigl[ \| Z_N(t+u)- Z_N(t) \|_B^2 \bigr]\leq \E^N_t[ \eta^T_N(h)] 
\quad\text{and}\quad\lim_{h\to0} \sup_{N}\E[ \eta_N(h)]=0,
\]
\end{itemize}
then the sequence $(Z_N)_{N \in \N}$ is tight in $\DD([0, \infty),B)$.
\end{prop}

In what follows, the first step consist in proving $(i)$ of Proposition~\ref{aldous} for $(\rho_N(t))_N$ and $(\trho_N(t))_N$ in $\SSS_p$, for any $t \ge 0$ and $p \ge 1$. Step 2 consists in showing the tightness $(\trho_N)_N$ in $\DD([0,\infty),\SSS_2)$ by proving  $(ii)$ of Proposition~\ref{aldous}, and Step 3 consists in extending this to any $\DD([0,\infty),\SSS_p)$ with $p >1$. Then, Step 4 goes back to $(\rho_N)_N$ and give its tightness in $\DD([0,\infty),\SSS_p)$ for any $p >1$. We conclude this section by proving some properties of the accumulation points. 

\medskip
{\bf Step 1: Tightness of $(\rho_N(t))_N$ in $\SSS_p$ for any $t\in [0,T]$ and $p\in[1,2]$.} For $p=1$, this is a consequence of Lemma \ref{prop:compS1}, which state that, for any $M>0$, the the set 
\[
\KK_M := \bigl\{ \rho \in \SSS_1^+ :  \; \|\rho \|_{\SSS_1} =1 \text{ and } 
\|X \rho X \|_{\SSS_1} + \bigl\| (i\nabla) \rho (i\nabla) \bigr\|_{\SSS_1} \le M
\bigr\}
\]
is compact in $\SSS_1$. Then, 
\[\begin{split}
 \mathbb{P}\big( \rho^N_t\not\in \KK_M\big)&\leq \mathbb{P}\big( \|X \rho^N_tX\|_{\SSS_1}+ \|\inab \rho^N_t\inab\|_{\SSS_1}>M\big)
+  \mathbb{P}\bigl( \|\rho \|_{\SSS_1} \neq 1  \bigr)
 \\
 &\leq \frac{1}{M}\E\big[\Tr(X\rho^N_tX)+\Tr(\inab \rho^N_t \inab)\big] 
 \leq \frac{K(t,\rho_0,\alpha,\beta)}{M},
 \end{split} \]
thanks to Proposition~\ref{propSPDE} and that $\Tr \rho^N_t=1$ with probability one for any $t\geq 0$. For $p\in(1,2]$, we just have to remark that $\KK_M$ is also compact in $\SSS_p$ for any $p\in(1,+\infty)$ since we have $\|\cdot\|_{\SSS_p}\leq \|\cdot\|_{\SSS_1}$.
 
Finally, using that $T_t$ is an isometry on $\SSS_p$ (see Proposition \ref{prop:Schatten}) together with the mapping theorem \cite[Theorem 2.7 pp. 21]{billingsley}, $(\trho_N(t))_N$ is also tight in $\SSS_p$.

\medskip
{\bf Step 2: Tightness of $(\trho_N)_N$ in $\DD([0,\infty),\SSS_2)$.}  Here, we just have to prove $(ii)$ of Proposition~\ref{aldous}. To do so, rewriting \eqref{SPDEpoisson} between $t$ and $t+h$ and then taking its $\SSS_2$-norm, we have 
\[\begin{split} 
 \bigl\| \trho^N_{t+h}  - \trho^N_t \bigr\|_{\SSS_2} \le   \alpha^2 & \int_t^{t+h}  
 \Bigl(  1  +  \Bigl\| \theta_\infty\bigl[ T_s [\trho^N_s] \bigr]  \Bigr\|_{\SSS_2} \Bigr) \,ds\\
 &+  \biggl\| \intPh  \bigl( T_{-s}  I_{p,x}\bigl[ T_s[ \trho^N_{s^-}]\bigr]-\trho^N_{s^-} \bigr) \tilde P_N(ds,dx,dp) \biggr\|_{\SSS_2}.
\end{split}\]
Using that $\| \trho^N_t \|_{\SSS_2} \le \| \trho^N_t \|_{\SSS_1} =\| \rho^N_t \|_{\SSS_1} =1$ at any time, since $T_t$ preserves all the Schatten norms, together with $(iii)$ of Proposition~\ref{prop:Schatten}, the control of the first integral in the r.h.s is easy: 
\begin {align*}
\int_t^{t+h}  
 \Bigl(  1  +  \Bigl\| \theta_\infty\bigl[ T_s (\trho^N_s) \bigr]  \Bigr\|_{\SSS_2} \Bigr) \,ds \le 
 h \biggl(1 + \frac1{\sqrt{2\pi}} \bigl\|\widehat \theta_\infty \bigr\|_1 \biggr) = 2 h.
\end{align*}
The second integral in the r.h.s. is an integral with respect to a random Poisson measure, so that by Lemma \ref{lem:MG} in Appendix \ref{sec:MGS2} and Lemma~\ref{lem:I_bound}, it comes 
\[\begin{split}
\E^N_t \Biggl[ \biggl\| \intPh \bigl(  T_{-s} I_{p,x}\bigl[ &T_s[ \trho^N_{s^-}]\bigr]-\trho^N_{s^-} \bigr) \tilde P_N(ds,dx,dp) \biggr\|_{\SSS_2}^2 \Biggr] \\
& =  \frac N {2R} \int_t^{t+h} \!\!\! \int_{-R}^R \!\int_\R
\Bigl\| T_{-s}\Bigl[  I_{p,x}\bigl[ T_s[ \trho^N_{s^-}]\bigr]  -T_s[\trho^N_{s^-}] \Bigr]  \Bigr\|_{\SSS_2}^2 \,ds \,dx \,\mu_m(dp), \\
& \le  \frac {2N}{R} \int_t^{t+h} \!\!\! \int_{-R}^R \!\int_\R
\Bigl(  \frac{\alpha^2}N + \frac{\alpha}{\sqrt N} e^{-2 p^2} \Bigr)^2 \,ds \,dx \,\mu_m(dp), \\
& \le  C\,N h,
\end{split}\]
where we have used again that $\| \trho^N_t \|_{\SSS_2} \le \| \rho^N_t \|_{\SSS_1} =1$ for all $t\geq 0$, and the constant $C$ depends on $\alpha$ and $\beta_0$. As a result, we obtain 
\begin{equation} \label{eq:tight1}
\E^N_t \Bigl[ \bigl\| \trho^N_{t+h}  - \trho^N_t \bigr\|_{\SSS_2} \Bigr] \le K \,h^{1/2},
\end{equation}
for some $K >0$ independent of $N$ and then $(ii)$ of Proposition~\ref{aldous} holds with $\eta_N^T(h) = K h^{1/2}$.  

\medskip
{\bf Step 3: Tightness of $(\trho_N)_N$ in $\DD([0,T],\SSS_p)$ for any $p>1$.}
the case $p\ge 2$ is clear since $\| \cdot\|_{\SSS_p} \le \| \cdot\|_{\SSS_2}$ in that case. 

For the remaining cases, we use the following standard H\"older inequality: 
for all $1\leq q \leq r \leq +\infty $, $p \in[q,r]$, and any compact operator $\rho$, we have
\[  \| \rho \|_{\SSS_p}\leq \|\rho\|^{\theta}_{\SSS_q} \|\rho\|^{1-\theta}_{\SSS_r}, 
\qquad\text{where}\qquad \frac{1}{p} =\frac{1-\theta}{r}+\frac{\theta}{q}.\]
We apply it for  $p \in (1,2]$, $q=1$ and $r=2$, so that $\theta= \frac2p -1$. Using the simple bound $\bigl\| \trho^N_{t+h}  - \trho^N_t \bigr\|_{\SSS_1} \le 2$ and the bound~\eqref{eq:tight1},
\begin{align}
\E^N_t \Bigl[ \bigl\| \trho^N_{t+h}  - \trho^N_t \bigr\|_{\SSS_p} \Bigr]
& \le 
\E^N_t \Bigl[ \bigl\| \trho^N_{t+h}  - \trho^N_t \bigr\|_{\SSS_1}^{1-\theta}
 \bigl\| \trho^N_{t+h}  - \trho^N_t \bigr\|_{\SSS_2}^{\theta}
 \Bigr] \nonumber \\
&  \label{trho_tight} \le 
\E^N_t \Bigl[ \bigl\| \trho^N_{t+h}  - \trho^N_t \bigr\|_{\SSS_1} \Bigr]^{1-\theta}
\E^N_t \Bigl[  \bigl\| \trho^N_{t+h}  - \trho^N_t \bigr\|_{\SSS_2} \Bigr]^{\theta}
 \le 2^{1-\theta}K^\theta h^{\theta/2}.
\end{align}

\medskip
{\bf Step 4: Tightness of $(\rho_N)_N$ in $\DD([0,\infty),\SSS_p)$ for any $p > 1$.} Using that the action of an isometry preserves the $\SSS_2$-norm and 
Lemma~\ref{rho_and_T} of Appendix~\ref{Appendix:B}, one has:
\[
\bigl\| \rho^N_{t+h}  - \rho^N_t \bigr\|_{\SSS_2} \leq \bigl\| \rho^N_{t+h}  - T_h[\rho^N_t] \bigr\|_{\SSS_2} +\bigl\| T_h[ \rho^N_t] -\rho^N_{t} \bigr\|_{\SSS_p} \leq  \bigl\| \trho^N_{t+h}  - \trho^N_t \bigr\|_{\SSS_2} +2 h^{1/2} \Tr(\inab \rho^N_t \inab).
\]
We fix $T >0$, $N \in \N^*$ and $h \in (0,1)$, and we define the random variable
\[
\eta^T_N(h) := \Bigl(   K + 2 \sup_{t \le T} \Tr(\inab \rho^N_t \inab)  \Bigr)  \sqrt h.
\]
Then from the above inequality, and the bound~\eqref{eq:tight1}, one has 
\[
\E^N_t \Bigl[ \bigl\| \rho^N_{t+h}  - \rho^N_t \bigr\|_{\SSS_2} \Bigr]
 \le  \E^N_t \Bigl[ \bigl\| \trho^N_{t+h}  - \trho^N_t \bigr\|_{\SSS_2} \Bigr]
+ 2 \sqrt h \, \E^N_t \Bigl[  \Tr(\inab \rho^N_t \inab) \Bigr]
\le
\E^N_t \Bigl[ \eta^T_N(h) \Bigr].
\]
But Proposition~\ref{unifkinetic} implies that $\lim_{h \to 0} \sup_N \E \bigl[ \eta^T_N(h) \bigr] = 0$, which means that the point $(ii)$ of Proposition~\ref{aldous} is satisfied. So the tightness holds for $p=2$. 
Then the tightness extend to all $p>1$ exactly has in the previous step.

\medskip
\paragraph{\bf Properties for the possible limits.}
 
Up to some extraction, we may assume that both $\trho^N$ and $\rho^N$ converges in law in $\DD([0,+\infty),\SSS_p)$ respectively to processes denoted by $\trho^\infty$ and $\rho^\infty$. We prove now the following useful result on the limit processes.
\begin{lem}\label{apriori}
The possible limit coupled processes $\trho^\infty$ and $\rho^\infty$ satisfy:
\begin{itemize}
\item[i)] $\trho^\infty$ and $\rho^\infty$ have continuous trajectories with respect to the $\SSS_p$-norm, for $p >1$;
\item[ii)] for any $t \ge 0$, $\rho^N_t$ and $\trho^N_t$ $\SSS_p$-converge in law respectively towards $\rho^\infty_t$ and $\trho^\infty_t$. Moreover, $\trho^\infty_t = T_t [\rho^\infty_t]$ for all time $t \ge 0$, almost surely;
\item[iii)] for any $t \ge 0$, $\rho^\infty_t$ and $\trho^\infty_t$ are symmetric nonnegative operators with trace one, a.s.;
\item[iv)] the kinetic energy and the momentum in position satisfy
\begin{equation} \label{mom_lim}
\E \Bigl[  \Tr ( \inab \rho_t^\infty \inab \bigr) \Bigr]  \le 
\Tr \bigl( \inab \rho_0 \inab \bigr) + \frac{4 \alpha^2}{\beta_0} t, \qquad
\E \Bigl[  \Tr ( X \rho_t^\infty X \bigr) \Bigr]  \le  C (1 + t^3),
\end{equation}
for some constant $C$ depending on $\rho_0$, $\alpha$ and $\beta_0$. 
\end{itemize}
\end{lem}

\begin{proof}
\emph{Item i)} is a standard consequence of the fact that the maximum size of jumps goes to zero as $N$ goes to infinity. Here, Lemma~\ref{lem:I_bound} yields for $p \ge 1$
\begin{equation}\label{conti}
\E\Big[\sup_{t\geq 0}\|\rho^N(t)-\rho^N(t^-)\|_{\SSS_p}\Big] \le 
\E\Big[\sup_{t\geq 0}\|\rho^N(t)-\rho^N(t^-)\|_{\SSS_1}\Big] \leq \frac C {\sqrt N},
\end{equation}
so that \cite[Theorem 13.4 pp. 142]{billingsley} implies that $\rho^\infty$, and also $\trho^\infty$, have $\SSS_p$-continuous trajectories for $p>1$ ($p=1$ not included since the convergence does not holds in that case).

\medskip
\emph{Item ii).} Since the limit processes have continuous trajectories,  we may pass to the limit at any fixed time $t$ (a thing that the Skorokhod topology do not necessary allows in the general case, see \cite[Section 13 pp. 138]{billingsley}): $\rho^N_t$ and $\trho^N_t$ both $\SSS_p$-converge in law towards respectively $\rho^\infty(t)$ and $\trho^\infty(t)$ for $p \in (1,2]$. Now,  for any fixed $t \ge 0$, by the $\SSS_p$-continuity of the operator $T_t$, we have
\[
\trho^\infty_t =T_t \bigl[ \rho^\infty_t\bigr], \quad \text{a.s.},
\]
and then we can exchange the ``almost sure'' and the ``for all $t \ge 0$'', since  $t\mapsto T_t \bigl[ \rho^\infty_t\bigr]$ taking its values in $\SSS_p$ is continuous. 

\medskip
\emph{Item iii).} The fact that the limit process is made of symmetric and nonnegative operators is straightforward. It remains to prove the trace property. Since the sets $\KK_M$ defined in Appendix \ref{sec:App_comp} by~\eqref{def:KM} are compact in $\SSS_1$, they are also compact and closed in the weaker topologies $\SSS_p$, for $p \in [1, +\infty]$. It implies, using again the Portmanteau theorem, that for any $t \ge 0$
\[
\Pro \Bigl( \rho^\infty_t \in \KK_M \Bigr) \ge \limsup_{N \to \infty} \Pro \Bigl( \rho^N_t \in \KK_M \Bigr) \ge 1  - \frac C M,
\]
by Proposition~\ref{propSPDE}. Since $\bigcup_{M>0} \KK_M \subset \SSS_1$, it implies that
$\Pro \bigl( \| \rho^\infty_t\|_{\SSS_1} =1 \bigr) =1$, and the same holds for $\trho^\infty_t$. 

\medskip
\emph{Item iv).} The inequalities~\eqref{mom_lim} on the kinetic energy and the moment in position for the limit processs $\rho^\infty$ simply follows from the same inequalities for $\rho^N$ proved in Proposition~\ref{propSPDE}, and from the fact that the kinetic energy and the position momentum can only decrease in the limit $N\to+\infty$ by Lemma~\ref{lem:exchange}.
\end{proof}

\subsection{Identification of all subsequence limits} \label{identificationsec}

This section is devoted to the identification of all the accumulation points by mean of a well posed stochastic differential equation on $\SSS_2$. For the sake of simplicity, throughout this section, we still denote by $\trho^N$ a converging subsequence in law in $\DD([0,+\infty),\SSS_2)$, with limit $\trho^{\infty}$.

Let us introduce for any $\rho\in\DD([0,+\infty),\SSS_2)$ the process
\begin{equation} \label{def:M}
M_t(\rho):=  \rho_t - \rho_0 
 + \alpha^2 \int_0^t  \Bigl( \rho_s -  T_{-s} \bigl[\theta_\infty \bigl[T_s[\rho_s]\bigr] \bigr] \Bigr) \,ds.
\end{equation}
Going back to \eqref{VN-PPPcomp}, we have
\begin{equation} \label{def:MN}
M_t ( \trho^N_t ) = \intP \Bigl(  T_{-s} \Bigl[I^N_{p,x} \bigl[ T_s\bigl[\trho^N_{s^-}\bigr] \bigr] \Bigr] -\trho^N_{s^-} \Bigl) \tilde P_N(ds,dx,dp) 
\end{equation}
which is a square-integrable martingale in $\SSS_2$ according to Lemma~\ref{lem:I_bound} and Lemma~\ref{lem:MG}. 
That last lemma also implies that the quadratic variation of $M( \trho^N)$ is given by
\[ \la\la \, M( \trho^N)\, \ra\ra (t,U,V)
=\int_0^t \big<A^N(s,\trho^N_s)U ,V \big>_{\SSS_2}ds,\]
with
\begin{multline}\label{defAN}
 \big<A^N(s,\trho^N_s)U,V \big>_{\SSS_2}=\frac{N}{2R} \int_{-R}^R \int_{\R} \Big\la  T_{-s} \Bigl[I^N_{p,x} \bigl[ T_s\bigl[\trho^N_{s}\bigr] \bigr] \Bigr] -\trho^N_{s} ,U\Big\ra_{\SSS_2}  \\
 \times \Big\la V,  T_{-s} \Bigl[I^N_{p,x} \bigl[ T_s\bigl[\trho^N_{s}\bigr] \bigr] \Bigr] -\trho^N_{s} \Big\ra_{\SSS_2} dx\mu_m(dp).
\end{multline}
for any $U,V\in \SSS_2$. From this we can prove the following result.
\begin{prop}\label{prop:varquad}
For any accumulation point $\trho^\infty$, $M(\trho^\infty)$ is a martingale with quadratic variation given by
\[ \la\la \, M( \trho^\infty)\, \ra\ra (t,U,V)=\int_0^t \big<A\big(s,\trho^\infty_s\big)(U),V \big>_{\SSS_2}ds,\]
where for any $U,V, \rho \in \SSS_2$,
\begin{multline}\label{defA}
 \big<A\big(s,\rho\big)U,V \big>_{\SSS_2}= \frac{\alpha^2}{2R}\sqrt{\frac {\beta_0}{8 - \beta_0}} \int_{-R}^R   \Bigl\la T_{-s} \Bigl[ \bigl[ i\gamma (\cdot -x),T_s[\rho] \bigr] \Bigr] ,U \Bigr\ra_{\SSS_2} \\
\times \Big\la V, T_{-s} \Bigl[ \bigl[ i\gamma (\cdot -x),T_s[\rho] \bigr] \Bigr]\Big\ra_{\SSS_2}  \,  dx.
\end{multline}

\end{prop}
The proof of that proposition is postponed to the end of that section. 

From now one, we denote $\kappa := \frac \alpha{\sqrt{2R}} \Bigl(\frac {\beta_0}{8 - \beta_0} \Bigr)^{1/4}$. 
We define for any $s \ge 0$ and $\rho \in \SSS_2$, a bounded linear operator  $\sigma(s,\rho) : L^2(-R,R) \to \SSS_2$ by
\[
\sigma(s, \rho) u :=  \kappa \int_{-R}^R  T_{-s} \Bigl[ \big[ \gamma (\cdot -x),T_s[\rho] \big] \Bigr]  u(x)  \,dx. 
\] 
Then its adjoint $\sigma^*(s,\rho) : \SSS_2 \to L^2(-R,R)$ is defined by
\[
\sigma^*(s, \rho) U := \biggl( x \mapsto \Bigl\la T_{-s} \Bigl[ \big[ \gamma (\cdot -x),T_s[\rho] \big] \Bigr],  U  \Bigr\ra_{\SSS_2}  \biggr),
\]
and $\sigma$ is a ``square root'' of $A$:
\[
\sigma(s,\rho) \sigma^*(s, \rho) = A(s,\rho).
\]

As a result, using a ``standard'' martingale representation theorem (see \cite[Theorem 8.2]{daprato} for instance), one can show that there exists a Brownian motion $B$ on $L^2(-R,R)$ with identity has covariance operator, possibly defined on an extension of the probability space, and such that almost surely:
\begin{equation}\label{EDS}\begin{split}
 M_t(\trho^\infty)= \trho^{\infty}_t -& \trho^\infty_0 
 + \alpha^2 \int_0^t  \Bigl( \trho^{\infty}_s -  T_{-s} \bigl[\theta_\infty \bigl[T_s[\trho^{\infty}_s]\bigr] \bigr] \Bigr) \,ds\\
 &+ i \kappa \int_0^{t} \int_{-R}^R  T_{-s} \Bigl[ \big[\gamma (\cdot -x),T_s[\trho^\infty_s] \big] \Bigr]  dB_s(x)dx \qquad \forall t\geq 0,
\end{split}\end{equation}

In order to apply that theorem, we need to check that $\Phi(t) := \sigma(t,\trho^\infty_t)$  is a predictable process with value in $L_2\bigl( L^2(-R,R), \SSS_2)$ the space of Hilbert-Schmidt operator from $L^2(-R,R)$ into $\SSS_2$. But $\Phi$ is predictable since it depends only on the time $t$ and the value of predictable process $\trho^\infty$ at this time. It is also Hilbert-Schmidt since, if $\Tr_{\SSS_2}$ denotes now the trace of operators acting on $\SSS_2$
\begin{align*}
\Tr_{\SSS_2} \bigl( \Phi(t) \Phi^*(t)\bigr) &  = 
\Tr_{\SSS_2} \bigl( \sigma(t,\trho^\infty_t) \sigma^*(t, \trho^\infty_t)\bigr) =
 \Tr_{\SSS_2} A(t, \trho^\infty_t) \\
 & = \kappa^2 \int_{-R}^R \Bigl\| T_{-t} \Bigl[ \big[\gamma (\cdot -x),T_t[\trho^\infty_t] \big] \Bigr]  \Bigr\|_{\SSS_2}^2  \,dx \\
 & \le 4 \kappa^2  \| \gamma \|_\infty^2 \int_{-R}^R \bigl\| \trho^\infty_t \bigr\|_{\SSS_2}^2 \,dx  \le  8 \kappa^2 R \| \gamma \|_\infty^2.
\end{align*}   
where we have used the fact that $T_s$ is an isometry on $\SSS_2$, the fact that $\|\cdot  \|_{\SSS_2} \le \| \cdot \|_{\SSS_1}$, and point $(iv)$ of Proposition~\ref{prop:Schatten}.

\medskip

The uniqueness for \eqref{EDS} is a direct consequence of  \cite[Theorem III.2.2]{yor}, which apply thanks to the following linear growth relations
\[\sup_{s\geq 0}\Big \| T_{-s}\bigl[ \bigl[\gamma_\infty, T_s[\rho] \bigr]\bigr]\Bigr \|_{\SSS_2}\leq C_2 \|\rho\|_{\SSS_2}\qquad
\text{and}\qquad\sup_{s\geq 0}\Bigl\| \rho-T_{-s}\bigl[\theta_\infty[T_s[\rho]]\bigr] \Bigr\|_{\SSS_2}\leq C_2 \|\rho\|_{\SSS_2}.\]
This characterizes uniquely the laws of all the accumulation points, and then prove the convergence in law.  

\medskip
Now, setting
\[
W_t(X)=  \kappa \int_0^t  \int_{-R}^R  \gamma (X-x) dB_s(x) \, dx ,
\]
we obtain a Brownian field on $\R$ with correlation function \eqref{eq:cor}, and also
\[ 
 \trho^{\infty}_t - \trho^\infty_0 
 + \alpha^2 \int_0^t  \Bigl( \trho^{\infty}_s -  T_{-s} \bigl[\theta_\infty \bigl[T_s[\trho^{\infty}_s]\bigr] \bigr] \Bigr) \,ds = i \int  T_{-s} \Bigl[ \bigl[  dW_s,T_s[\trho^\infty_s] \bigr] \Bigr],
\]
Finally, to conclude the proof, we apply It\^o's formula~\cite[Theorem III.1.3]{yor} to $\rho^\infty_t=T_t[\trho^\infty_t]$.

\begin{proof}[of Proposition \ref{prop:varquad}] 
Since  $M(\trho^N)$ is a martingale,
\begin{equation}\label{MGN}
\mathbb{E}\Big[ \Psi\big(\la M_{s_j}(\trho^N),U_j\ra_{\SSS_2},1\leq j \leq n\big)\Big(\la M_t(\trho^N)-M_s(\trho^N),U\ra_{\SSS_2}\Big) \Big]=0,
\end{equation}
for all $N,n\geq 1$, bounded continuous function $\Psi$, every sequence $0< s_1<\cdots <s_n \leq s <t$, and every family $(U_j)_{j\in \{1,\dots,n\}},U\in \SSS_2$. 

From the definition~\eqref{def:M}, it comes also that $\rho\mapsto M(\rho)$ is a bounded continuous function on $\DD([0,+\infty),\tilde{\SSS}_2)$ according to Proposition~\ref{prop:Schatten}, where $\tilde{\SSS}_2$  stands for the unit ball of $\SSS_2$. This is in fact a standard result in the Skorokhod topology. So  we can pass to the limit in $N$ in \eqref{MGN} and then obtain that $M(\trho^\infty)$ is a martingale (a bounded martingale according to Lemma~\ref{apriori}). 

To obtain the quadratic variation of $M(\trho^\infty)$, let us first remark that for any $\rho \in \SSS_2$, we have
\[
I_{p,x}^N[\rho] - \rho \mp \frac {i\alpha} {\sqrt N} e^{-2p^2} \bigl[ \gamma(\cdot - x) , \rho\bigr] =  R^N_{p,x} [\rho ],
\]
\[\text{where} \qquad R^N_{p,x}[\rho]  := \frac{\alpha^2}{N} \bigl(\theta_p[\rho] - \rho \bigr) \pm \frac{i\alpha}{\sqrt{N}}\Big(\sqrt{1-\frac{\alpha^2}{N}}-1\Big)\bigl[ \gamma(\cdot - x) , \rho\bigr] \]
and then according to Proposition~\ref{prop:Schatten}
\[
\bigl\| R^N_{p,x} [\rho] \bigr\|_{\SSS_2} \le \frac C N \| \rho \|_{\SSS_2},
\qquad 
\bigl\| I_{p,x}^N[\rho] - \rho [\rho] \bigr\|_{\SSS_2}
\le \frac C {\sqrt N} \| \rho \|_{\SSS_2}
\]
for some numerical constant $C>0$. Now from the definitions~\eqref{defAN} and~\eqref{defA} of $A_N$ and $A$, it comes that for any $s \in \R^+$, $\rho ,U,V \in \SSS_2$,
\begin{align*}
\bigl| A^N(s,\rho)(U,V) - A(s,\rho)(U,V) \bigr| &  \le 
\frac {C^2} {R \sqrt N}  \| U \|_{\SSS_2} \| V \|_{\SSS_2} \!\! \int_{-R}^R \int_{\R} \bigl\| R^N_{p,x} [\rho] \bigr\|_{\SSS_2}  \bigl\| I_{p,x}^N[\rho] - \rho [\rho] \bigr\|_{\SSS_2}  dx \, \mu_m(dp),\\
& \le \frac{C'}{\sqrt N} \| U \|_{\SSS_2} \| V \|_{\SSS_2} \| \rho \|_{\SSS_2}^2.
\end{align*}
Finally, using that the process
\[
t\mapsto \la M_{t}(\trho^N),U \ra^2_{\SSS_2}-\int_0^t \big<A^N(s,\trho^N_s)(U),U \big>_{\SSS_2}ds
\] 
is a martingale, and 
by the continuity of $\rho \mapsto M(\rho)$ and $\rho \mapsto \int_0^t \la A(s,\rho_s)U,V \ra_{\SSS_2} \,ds$ on $\DD([0,\infty),\SSS_2)$, we can argue as in~\eqref{MGN} and get that
\[
t\mapsto \la M_{t}(\trho^\infty),U \ra^2_{\SSS_2}-\int_0^t \big<A(s,\trho^\infty_s)(U),U \big>_{\SSS_2}ds
\]  
is also a martingale. This concludes the proof.
\end{proof}

\section*{Appendix}
\appendix

\section{Bounded operators in the Schatten class} \label{Sch_Bd}

In this section, we prove some simple inequalities for the Schatten norms that are extensively used throughout this paper. Let us start with the following two simple estimates.
\begin{lem} \label{dist_op}
Let $\phi, \psi \in L^2(\R^d)$. For all $p \in [1, \infty]$, we have :
\begin{enumerate}
\item[(i)]
\[
\bigl\| |\phi \ra \la \psi| \bigr\|_{\SSS_p}  = \| \phi \|_2 \,  \| \psi\|_2;
\]
\item[(ii)] if moreover $\phi,\psi$ are unit vectors of $L^2(\R^d)$, we have
\[
\bigl\| |\phi \ra \la \phi|- |\psi \ra \la \psi|\bigr\|_{\SSS_p} 
= 2^{1/p} \sqrt{1 - \bigl| \la  \phi | \psi \ra \bigr|^2} \le 2^{1/p} \| \phi - \psi \|_2.
\] 
\end{enumerate}
\end{lem}

\begin{proof} \emph{Item (i).} First, setting $\rho := |\phi \ra \la \psi|$ it is simple to see that $\rho^\ast \rho = \|\phi\|^2_2 |\psi \ra \la \psi|$, so that the calculations of the $\SSS_p$-norms of $\rho$ are straightforward.

\emph{Item (ii).} Defining now $\rho := |\phi \ra \la \phi|- |\psi \ra \la \psi|$ and $a:=\la  \psi | \phi \ra$, we have
\[
\rho^\ast \rho = |\phi \ra \la \phi| + |\psi \ra \la \psi|
- a \, |\psi \ra \la \phi|
- \overline{a} \, |\phi \ra \la \psi|.
\]
Then, it is not difficult to see that 
$ \rho^\ast \rho | \phi \ra = \bigl( 1 - |a|^2 \bigr) | \phi \ra$ and 
$ \rho^\ast \rho | \psi \ra = \bigl( 1 - |a|^2 \bigr) | \psi \ra$, so that $\rho^\ast \rho = \bigl(1 -|a|^2 \bigr) P_2$, with $P_2$ the orthogonal projection on the two dimensional subspace generated by $(\psi, \phi)$, when $\psi \neq \pm \phi$. This yields the equality.

Now, to obtain the bound, let us remark that the assumption $\| \psi \|_2 = \bigl\| \phi + ( \psi - \phi) \bigr\|_2= 1$ implies that $2 \mathrm{Re } \la \phi - \psi | \phi \ra =  \| \psi - \phi \|_2^2$.
Therefore, we have 
\begin{align*}
1 - \bigl| \la  \phi | \psi \ra \bigr|^2 & = 
1 - \bigl| 1 + \la  \phi | (\psi - \phi) \ra \bigr|^2 
= 2 \mathrm{Re} \la \phi - \psi | \phi \ra - \bigl| \la  \phi | (\psi - \phi) \ra \bigr|^2 \\
& = \| \psi - \phi \|_2^2 -  \bigl| \la  \phi | (\psi - \phi) \ra \bigr|^2
\le \| \psi - \phi \|_2^2,
\end{align*}
and the conclusion follows.
\end{proof}

\begin{lem}\label{lem:mix}
Assume that $\rho : L^2(\R^d) \mapsto L^2(\R^d)$ is a symmetric nonnegative operator, and that $g$ is a continuous function on $\R^d$. Then
\[
\Tr\bigl( \inab \rho g(X) + g(X) \rho \inab \bigr)  \le 2 \sqrt{\Tr \bigl( g(X) \rho g(X)) \, \Tr \bigl( \inab \rho  \inab \bigr)}
\]
\end{lem}
\begin{proof}
This can be shown using the standard decomposition of symmetric nonnegative operator $\rho = \sum_{i = 1}^{+\infty} \lambda_i |\psi_i \ra \la \psi_i |$, where $(\lambda_i)_{i \in \N^\ast}$ is the non-increasing sequence of eigenvalues of $\rho$, and $(\psi_i)_{i \in \N^\ast}$ is an associated orthogonal family of eigenvectors. Hence,
\[\Tr\bigl( \inab \rho g(X) + g(X) \rho \inab \bigr)  = 
\sum_{j =1}^{+\infty} 2 \lambda_j  \textrm{Re}  \bigl\la g(X) \psi_j  \big| \inab \psi_j \bigr\ra , 
\]
and then
\begin{align}
\biggl| \Tr\bigl( \inab \rho g(X) + g(X) \rho \inab \bigr) \biggr| & \le
2 \sum_{j =1}^{+\infty} \lambda_j \| g(X) \psi_j \|_2 \| \nabla \psi_j \|_2 \nonumber \\
& \le 2 \, \biggl( \sum_{j =1}^{+\infty} \lambda_j \| g(X) \psi_j \|_2^2 \biggr)^{1/2}
\biggl( \sum_{j =1}^{+\infty} \lambda_j \| \nabla \psi_j \|_2^2 \biggr)^{1/2} \label{tracebound}\\
& \le 2 \, \sqrt{\Tr \bigl( g(X) \rho g(X)) \, \Tr \bigl( \inab \rho \inab \bigr)}\nonumber.
\end{align}
\end{proof}

In the following proposition we collect simple inequalities for the Schatten norm. More specifically, we are interested by two operations on compact operators:
\begin{equation} \label{op_theta}
\rho \mapsto \theta[\rho] := \frac1{\sqrt{2\pi}} \int_\R e^{ikX} \rho e^{-ikX}  \widehat \theta(k) \,dk ,
\qquad
\rho \mapsto i \bigl(\gamma \rho - \rho \gamma \bigr),
\end{equation}
for function $\theta$ and $\gamma:\R \rightarrow \R$. Below, the Fourier convention is the following
\[\widehat{f}(k)=\frac{1}{\sqrt{2\pi}}\int f(X)e^{-ikX}dX\qquad\text{and then }\qquad f(X)=\frac{1}{\sqrt{2\pi}}\int \widehat{f}(k)e^{ikX}dk.\]

\begin{prop} \label{prop:Schatten}
We have for any compact operator $\rho$ on $L^2(\R)$, and any $p \in [1, \infty]$ :
\begin{itemize}
\item[(i)] $ \|\rho^*  \|_{\SSS_p} = \| \rho \|_{\SSS_p}$;
\item[(ii)] for any isometry $U$ on $L^2(\R)$,  $\| U \rho U^*  \|_{\SSS_p} = \| \rho \|_{\SSS_p}$;
\item[(iii)] if $\theta$ has an integrable Fourier transform, then the operator defined by~\eqref{op_theta} satisfies:
\[\; \bigl\|   \theta[\rho]  \bigr\|_{\SSS_p} \le \frac1{\sqrt{2\pi}} \bigl\| \widehat \theta\bigr\|_1 \| \rho \|_{\SSS_p};
\]
\item[(iv)] if $\gamma$ is bounded,
$ \displaystyle \; \|   \gamma \rho \|_{\SSS_p} \le  \| \gamma\|_\infty \| \rho \|_{\SSS_p}$,
and
$ \displaystyle \;  \| \rho  \gamma \|_{\SSS_p} \le  \| \gamma\|_\infty \| \rho \|_{\SSS_p}. $
\end{itemize}
\end{prop}

\begin{proof}
\emph{Item (i).} It is an easy consequence of the expansion~\eqref{expansion}: it implies in fact that
\[
\rho^\ast = \sum_{i=1}^{+\infty} \mu_i(\rho) \, | \phi_i \ra \la \psi_i |,
\]
and so $\mu_i(\rho^\ast) = \mu_i(\rho)$ for each $i \ge 1$. 

\emph{Item (ii).} Using again the expansion \eqref{expansion}, we have
\[
U\rho U^\ast = \sum_{i=1}^{+\infty} \mu_i(\rho) \,  | U \psi_i \ra \la U \phi_n |,
\]
so that $\mu_i( U\rho U^\ast) = \mu_i(\rho)$ for every $i\geq 1$.

\emph{Item (iii).} The multiplication by $e^{ikX}$ is clearly an isometry of $L^2(\R)$, and then by $(ii)$ we get $\bigl\| e^{ikX} \rho  e^{-ikX} \|_{\SSS_p} =  \| \rho \|_{\SSS_p}$. Moreover, by \eqref{op_theta}, it comes
\[
\bigl\| \theta[\rho] \bigr\|_{\SSS_p} \le \frac1{\sqrt{2\pi}} \int_\R  \bigl\| e^{ikX} \rho e^{-ikX}  \bigr\|_{\SSS_p} \bigl|\widehat \theta(k) \bigr|  \,dk = 
\frac1{\sqrt{2\pi}} \bigl\| \widehat \theta \bigr\|_1 \| \rho \|_{\SSS_p}.
\]

\emph{Item (iv).} It is a simple consequence of the H\"older inequality for compact operator: $\| AB\|_{\SSS_p} \le \| A\|_{\SSS_\infty} \| B\|_{\SSS_p}$, that can be found for instance in~\cite[Theorem 2.8]{simon}. We apply it to $A=\gamma$ and $B= \rho$, and use the fact that the $\SSS_\infty$-norm of $\gamma$ acting as a multiplication operator on $L^2(\R)$  is lower than its infinite norm: $\| \gamma \|_{\SSS_\infty} \le \|\gamma \|_\infty$. Remark that the multiplication operator $\gamma $ is not compact but that $\gamma \rho$ and $\rho \gamma$ are when $\rho$ is. The conclusion follows.
\end{proof}

Let us finish this section with the following lemma providing equalities which are useful to prove the tightness in the proof of Theorem \ref{thm:main}.
\begin{lem}\label{kin_eq}
For any $\theta$ such that $\widehat  \theta \in L^1(\R)$, and any $\rho \in \SSS_1^+$, we have for the operator $\theta[\cdot]$ defined by~\eqref{op_theta}:
\begin{enumerate}
\item[(i)] 
\[
\Tr \bigl( X \theta[\rho] X \bigr) =  \theta(0) \Tr \bigl( X \rho X \bigr),
\]
\item[(ii)] if in addition $\widehat \theta (1+k^2) \in L^1(\R)$ and $\Tr \bigl(  \inab \rho \inab \bigr) <  \infty$, then 
\[
\Tr \bigl( \inab \theta[\rho] \inab \bigr) = 
\theta(0) \Tr \bigl( \inab \rho \inab \bigr) 
+i \,  \theta'(0) \Tr \bigl( \inab \rho + \rho \inab \bigr)
- \theta''(0) \Tr \rho.
\] 
\item[(iii)] if in addition $\widehat \theta (1+\vert k \vert) \in L^1(\R)$
and $ \Tr \bigl(  X \rho X \bigr)+ \Tr \bigl(  \inab \rho \inab \bigr) <  \infty$, then 
\[
\Tr \bigl( \inab \theta[\rho] X +X \theta[\rho] \inab \bigr) = 
\theta(0) \Tr \bigl( \inab \rho X +X \rho\inab \bigr)+ i \,  \theta'(0) \Tr \bigl(  \rho X+X \rho \bigr).
\]
\end{enumerate}
\end{lem}

\begin{proof} {\sl Item $(i)$.} Using first that for any $k \in \R$, multiplication by $e^{i k X }$ commutes with multiplication by $X$, we obtain using $(ii)$ of Proposition~\ref{prop:Schatten} that
\[
\Tr \bigl(  X e^{ikX} \rho  e^{-ikX} X \bigr) 
 = \Tr \bigl( e^{ikX} X \rho X e^{-ikX} \bigr) 
 = \Tr \bigl( X \rho X \bigr).
\]
The conclusion follows by integrating with respect to $(2\pi)^{-1/2} \widehat \theta(k) \,dk$, and that $\sqrt{2\pi} \, \theta(0) = \int_\R\widehat \theta(k) \,dk$.

\medskip
{\sl Item $(ii)$} 
Using that $\inab e^{ikX}= e^{ikX} \bigl[  \inab -k\bigr]$, we first get that
\begin{align*}
\Tr \bigl(  \inab e^{ikX} \rho  e^{ikX} \inab \bigr) 
& = \Tr \Bigl( e^{ikX} \bigl[  \inab -k\bigr]  \rho \bigl[  \inab -k\bigr] e^{-ikX} \Bigr) 
=  \Tr \Bigl( \bigl[  \inab -k\bigr]  \rho \bigl[  \inab -k\bigr]  \Bigr) \\
& =  \Tr \bigl(  \inab \rho  \inab \bigr) + k^2 \Tr \rho - k \Tr \bigl( \inab \rho + \rho \inab \bigr) 
\end{align*}
Multiplying by $(2\pi)^{-1/2} \widehat \theta(k)$ and summing up in $k$, it comes by definition~\eqref{op_theta} of $\theta[\rho]$:
\[
\sqrt{2\pi} \Tr \bigl( \inab \theta[\rho]  \inab \bigr)= \Tr \bigl(  \inab \rho  \inab \bigr) \int_\R \widehat \theta(k) \,dk  +  \Tr \rho \int_\R k^2 \widehat \theta(k) \,dk 
- \Tr \bigl( \inab \rho + \rho \inab \bigr)  \int_\R k\widehat \theta(k) \,dk,
\]
and the conclusion follows since under our hypothesis 
$\sqrt{2\pi} \, \theta(0) = \int_\R\widehat \theta(k) \,dk$, 
$\sqrt{2\pi} \,\theta'(0) = i \, \int_\R k\widehat \theta(k) \,dk$ and 
$\sqrt{2\pi} \, \theta''(0) = - \int_\R k^2 \widehat \theta(k) \,dk$.

Finally, the proof of item $(iii)$ follows the same line as the proofs of the two first items.
\end{proof}

\section{Free evolution of density operators}\label{Appendix:B}

In this section we investigate some properties involving the momentum in position, the kinetic energy, and the distance in $\SSS_1$ to the initial condition, for the solution of the free evolution 
\[i \partial_t \rho_t = \bigl[ H_0, \rho_t \bigr],\]
with initial condition $\rho_0\in \SSS_1^+$ satisfying $\Tr \rho_0=1$. 
\begin{lem}\label{free_case}
We have for all $t>0$
\[
\partial_t \Tr ( X \rho_t X) = \Tr\bigl( \inab \rho_t X + X \rho_t \inab \bigr)
\quad \text{and} \quad
\partial_t \Tr\bigl( \inab \rho_t X + X \rho_t \inab   \bigr) = 2 \Tr \bigl( \inab \rho_t \inab \bigr).
\]
\end{lem}

\begin{proof}
We proceed by diagonalization, that is we write
\begin{equation}\label{diag} \rho_t=\sum_{j=1}^{+\infty} \lambda_j  |\psi^j_t \ra \la \psi^j_t |\qquad \text{with} \qquad \rho_0=\sum_{j=1}^{+\infty} \lambda_j  |\psi^j_0 \ra \la \psi^j_0 |,\end{equation}
where for all $j\geq 1$, $\psi^j$ satisfies the free Schr\"odinger equation $i \partial_t \psi^j_t =H_0 \psi^j_t$ with initial condition $\psi^j_0$.
Here, we only have to obtain the associated equality for solution $\psi_t$ of the free Sch\"odinger equation $i\partial_t \psi_t = H_0 \psi_t$. It comes after some computations 
\begin{align*}
\partial_t \biggl( \int_\R X^2 | \psi_t(X)|^2 \,dX \biggr) & = 
2\, \mathrm{Im} \biggl( \int_\R X \nabla \psi_t(X) \overline{\psi_t(X)} \,dX \biggr),\\
\partial_t \mathrm{Im} \biggl( \int_\R X \nabla \psi_t(X) \overline{\psi_t(X)} \,dX \biggr) & = 2 \int_\R |\nabla \psi_t(X)|^2 \,dX,
\end{align*}
and the result for $\rho_t$ follows from \eqref{diag}.
\end{proof}

\begin{lem} \label{rho_and_T}
We have for all $t\geq 0$
\[
\bigl\| \rho_t - \rho_0 \bigr\|_{\SSS_1} \le 2 \sqrt { \Tr \bigl( \inab \rho_0 \inab \bigr)\,t }.
\]
\end{lem}

\begin{proof} 
Using as in the previous lemma the decomposition~\eqref{diag} and 
according to Lemma~\ref{dist_op}, we have for all $t\geq 0$
\[
\bigl\| \rho_t - \rho_0 \bigr\|_{\SSS_1} 
\le  2 \sum_{j=1} ^{+\infty} \lambda_j \| \psi^j_t - \psi^j_0 \|_2. 
\]
Now, using the Schr\"odinger equation, we also get for all $j\in \mathbb{N}^\ast$
\[\begin{split}
\frac d {dt} \|\psi^j_t -\psi^j_0\|^2_2  = 2 \, \mathrm{Re} \int_{\R}  \partial_t \psi^j_t(X) \overline{\psi^j_0(X)} \,dX  &= 
   \mathrm{Im} \int_{\R}  \nabla \psi^j_t (X)\cdot  \overline{\nabla \psi^j_0(X)} \,dX \\
  & \le \| \nabla \psi^j_t\|_2 \| \nabla \psi^j_0 \|_2 \le \| \nabla \psi^j_0 \|_2^2, 
\end{split}\]
since the kinetic energy $\| \nabla \psi^j_t \|_2^2$ is preserved by the free evolution. Hence,
\[
\bigl\| \rho_t - \rho_0 \bigr\|_{\SSS_1} 
\le 2  \sqrt t \sum_{j=1}^{+\infty} \lambda_j \| \nabla \psi_j \|_2 
=  2  \sqrt {  \Tr \bigl( \inab \rho_0 \inab \bigr)\,t }. 
\]
\end{proof}

\section{Martingale in the space of Hilbert-Schmidt operators.} \label{sec:MGS2}

In this section, $E$ is a Borel set of $\R^d$, and the nonnegative measure $\mu(ds, dv)$ is such that $\mu([0,t]\times E)<+\infty$ for any $t\geq 0$.
Let $(\Omega,\mathcal{F},\Pro)$ be a probability space on which we consider a Poisson random measure $P$ on $\R^+\times E$ with intensity $\mu(ds, dv)$.  Moreover, let us consider the filtration 
\[\mathcal{F}_t:=\sigma(P((0,s ],A), \,s\leq t, \, A\in \mathcal{B}(E)).\] 
Let also consider $F:[0,+\infty)\times E\times\Omega \to \SSS_2$ satisfying:
\begin{itemize}
\item for all $t\geq 0$, $(v,\omega)\mapsto F(t,v,\omega)$ is $\mathcal{B}(E)\otimes \mathcal{F}_{t^-}$- measurable,
\item for all $v\in E$ and $\omega\in\Omega$, $t\mapsto F(t,v,\omega)$ is left-continuous,
\item for all $T>0$ \[ \int_0^T \int \E[ \|F(t,v)\|^2_{\SSS_2}] \mu(dt,dv) <+\infty. \]
\end{itemize}

\begin{lem}\label{lem:MG}
Introducing the compensated Poisson random measure $\tilde{P}=P-\mu$, the process
\[M_t:=\int_0^t F(s,v)\tilde{P}(ds,dv),\]
is a is square-integrable martingale in $\SSS_2$, such that for any $\phi,\psi \in L^2(\R)$ the process $M_{\phi,\psi}:=\la \phi \vert M \vert \psi \ra$ is a square integrable martingale with quadratic variation 
\[
\la \la  M_{\phi,\psi} \ra\ra (t)= \int_0^t  \vert \la \phi \vert F(s,v) \vert \psi \ra\vert^2 \mu(ds,dv).
\]
Moreover, we also have for any $t,h\geq 0$,
\[ \E\bigl[ \| M_{t+h}-M_t \|^2_{\SSS_2} \big| \mathcal F_t \bigr]=\int_t^{t+h} \int \|F(s,v)\|^2_{\SSS_2} \mu(ds ,dv).\]
\end{lem}

\begin{proof} 
First, according to \cite[Theorem 4.2.3 pp. 224]{applebaum} it is direct that $M_{\phi,\psi}$ is a square-integrable scalar martingale for any $\phi,\psi\in L^2(\R)$, and then $M$ is a martingale in $\SSS_2$. Moreover, we also have that for any $t\geq 0$
\[\E\Big[\Big \vert \int_0^t \int  \la \phi \vert F(s,v)\vert \psi \ra \tilde{P}(ds,dv) \Big \vert^2 \Big]=\int^t_0 \int   \E[\vert  \la \phi \vert F(s,v)\vert \psi \ra\vert^2 ] \mu(ds,dv).\]
From this last relation, using that for all $h>0$ $P((t,t+h])$ is independent of $\mathcal{F}_t$ and standard properties of the conditional expectation, it is not difficult to see that the quadratic variation of $M_{\phi,\psi}$ is in fact given by $t\mapsto \int_0^t  \vert \la \phi \vert F(s,v) \vert \psi \ra\vert^2 \mu(ds,dv)$. To prove that $M$ is square-integrable in $\SSS_2$ let us remind the following relation
\[ \|\rho\|^2_{\SSS_2}=\sum_{j=1}^{+\infty} \la \psi_j \vert \rho^\ast\rho \vert \psi_j \ra = \sum_{j,l=1}^{+\infty} \vert \la \psi_j \vert \rho \vert \psi_l \ra \vert^2.\]
Hence, for any $t\geq 0$
\[\begin{split}
 \E \bigl[\|M_t\|^2_{\SSS_2} \bigr] & = \sum_{j,l=1}^{+\infty} \E\biggl[ \biggl\vert \int_0^t  \la \psi_j \vert F(s,v) \vert \psi_l \ra \tilde{P}(ds,dv) \biggr\vert^2\biggr]\\
 &=  \sum_{j,l=1}^{+\infty} \int_0^t   \E\big[ \vert \la \psi_j \vert F(s,v) \vert \psi_l \ra\vert^2 \big] \mu(ds,dv)
=  \int_0^t \int \E \bigl[ \|F(s,v) \|^2_{\SSS_2}\bigr]  \mu(ds,dv)<+\infty,
\end{split}\]
and in the same way
\[  
\E \bigl[ \| M_{t+h}-M_t \|^2_{\SSS_2} \big|  \mathcal F_t \bigr] 
=\int_t^{t+h} \int \|F(s,v)\|^2_{\SSS_2} \mu(ds ,dv).
\]
\end{proof}

%
%
%
\section{Compact sets in the trace class.} \label{sec:App_comp}

Here, we prove (and quote) the following helpful proposition regarding the compactness of subset of $\SSS_1$. The following result is used to prove the tightness in Theorem \ref{thm:main}.
\begin{prop} \label{prop:compS1}
Let $M>0$. The subset defined by
\begin{equation} \label{def:KM}
\KK_M := \bigl\{ \rho \in \SSS_1^+ :  \; \|\rho \|_{\SSS_1} =1 \text{ and } 
\|X \rho X \|_{\SSS_1} + \bigl\| (i\nabla) \rho (i\nabla) \bigr\|_{\SSS_1} \le M
\bigr\}
\end{equation}
is a compact in $\SSS_1$. 
\end{prop}
An interesting remark is that $\KK_M$ is also compact in $\SSS_p$ for all $p\in[1,+\infty)$ since $\|\cdot\|_{\SSS_p}\leq \|\cdot\|_{\SSS_1}$. The proof of Proposition~\ref{prop:compS1} use the following lemma.
\begin{lem} \label{lem:ortho}
For any $R>0$, we define the subset $K_R$ of $L^2(\R)$ by
\begin{equation}\label{def_KR}
K_R := \bigl\{  \psi \in H^1(\R) : \; \| \psi\|_2=1 \text{ and } \|X\, \psi \|_2^2 + \| \nabla \psi \|_2^2 \le R \bigr\}.
\end{equation}
There exists a positive integer $N(R)$ such that 
for any orthonormal family of vectors $(\psi_j)_{j \in J}$ that all belongs to $K_R$:
\[
\forall j \in J, \; \psi_j \in K_R,
\quad \text{and} \quad 
\forall j,l \in J, \; \la \psi_j, \psi_l \ra = \delta_{j,l}, 
\]
then the cardinal of the family is bounded by $N(R)$, i.e. $\mathrm{Card} \,I \le N(R)$.
\end{lem}
\begin{proof}[of the Lemma~\ref{lem:ortho}]
Here, we argue by contradiction. Fix $R>0$ and assume that for any $N$, we can construct a orthonormal family $(\psi_j^N)_{i \le N}$ of $K_R$. Using that $K_R$ is a compact subset of $L^2(\R)$, and then a diagonal extraction argument, we can find a subsequence $(\psi_j^{\sigma(N)})_{j \le  \sigma(N)}$ such that for any $j\in J$, the sequence $(\psi_j^{\sigma(N)})_{N \in \N}$ converges strongly in $L^2(\R)$ towards some $\psi_j$. Hence, $(\psi_j)_{j \in \N}$ forms a countable orthonormal family of $K_R$. This contradicts the compactness of $K_R$ in $L^2(\R)$, since the sequence $(\psi_j)_{j \in \N}$ cannot have any accumulation  point in $K_R$.
\end{proof}

\begin{proof}[of Proposition~\ref{prop:compS1}]
Let  $M>0$ and  $(\rho_n)_{n \in \N}$ be a sequence of operators in $\KK_M$. Let us start the proof by reminding the reader about some direct computations which are useful in what follows. Using the diagonalization of each operator $\rho_n$, we write
\[
\rho_n = \sum_{i=1}^{+\infty} \lambda_i^n | \psi_i^n \ra \la \psi^n_i |, \quad \text{with} \quad  \| \psi_n^i \|_2 =1,
\] 
where the (nonnegative) eigenvalues are listed in decreasing order: $\lambda^n_1 \ge \lambda_2^n  \ge \dots $. Hence,
\[
X \rho_n X = \sum_{i=1}^{+\infty} \lambda_i^n | X \psi_i^n \ra \la X \psi^n_i |,
\qquad 
\inab \rho_n \inab = 
\sum_{i =1}^{+\infty} \lambda_i^n | \inab \psi_i^n \ra \la \inab \psi^n_i |,
\]
which yields
\begin{equation} \label{trace_bd}
\| X \rho_n X \|_{\SSS_1} + \bigl\| \inab \rho_n \inab \|_{\SSS_1}  
= \Tr( X \rho_n X ) + \Tr \bigl( \inab \rho_n \inab \bigr) 
 = \sum_{i=1}^{+\infty} \lambda_i^n \bigl( \| X \psi_i^n \|_2^2 +  \| \nabla \psi_i^n \|_2^2 \bigr). 
\end{equation}

{\sl Step 1: The convergence $\sum_{i\geq 1} \lambda_i^n =1$ is uniform in $n \in \N$.} To prove it, we just have to show that for any $\ep >0$, there exists an $I \in \N$ such that $\sup_{n \in \N} \sum_{i \ge I} \lambda^n_i \le \ep$. We will show that is true for $I = N(R)$ with $R= M/ \ep$, where $N$ is the integer given by Lemma~ref{lem:ortho}: it bounds the cardinal of any orthonormal family made of elements of $K_R$. By~\eqref{trace_bd} we have for all $n \in \N$: 
\begin{align*}
M & \ge \| X \rho_n X \|_{\SSS_1} + \bigl\| \inab \rho_n \inab \|_{\SSS_1}  \ge R \sum_{\psi_i^n \notin K_R} \lambda_i^n 
 = R \biggl( 1 - \sum_{\psi_i^n \in K_R} \lambda_i^n  \biggr) \\
& \ge R \biggl( 1 - \sum_{i \le N(R)} \lambda_i^n  \biggr)
= R  \sum_{i \ge N(R)} \lambda_i^n
= \frac M \epsilon \sum_{i \ge I} \lambda_i^n,
\end{align*}
and the conclusion follows. 
Above, we pass from the first to the second line using Lemma~\ref{lem:ortho}, which implies that the cardinal of $\{ i, \; \psi_i^n \in K_R \}$ is bounded by $N(R)$. Then, the sum $\sum_{\psi_i^n \in K_R} \lambda_i^n$ is always smaller that $\sum_{i \le N(R)} \lambda_i^n $ since the sequence $(\lambda^n_i)_{i \in \N}$ is non-increasing.

\medskip
{\sl Step 2: A diagonal extraction.} Since the sequences $(\lambda_i^n)_{n \in \N}$ are all bounded by $1$, we may use a diagonal extraction argument, and then assume that for each $i \in \N^\ast$, $\lim_n\lambda_i^n = \lambda_i \in [0,1]$. The nonnegative $\lambda_i$ are also listed in decreasing order, $\lambda_1 \ge \lambda_2  \ge \ldots $, and they also satisfy $\sum_{i\geq 1} \lambda_i = 1$ thanks to the first step.

Let us consider now $n_0$ as the maximal integer for which $\lambda_i \neq 0$ ( $n_0 = +\infty$ if this never happens).
If $i \le n_0$, it is clear by~\eqref{trace_bd} that
\[
\limsup_{n \rightarrow + \infty} \| X \, \psi^n_i\|_2+ \| \nabla \psi^n_i\|_2 < + \infty.
\]
Moreover, since for any $R >0$ the subset $K_R$ defined by \eqref{def_KR} is compact in $L^2(\R)$, up to an diagonal extraction argument we may still assume that for any $i\in\mathbb{N}^\ast$ we have $ \| \psi^n_i -\psi_i \|_2 \rightarrow 0$, for some $\psi_i$ in $L^2(\R)$ with $\| \psi_i \|_2 = 1$. 

\medskip
{\sl Step 3: Conclusion.} Let us define $\rho := \sum_{i =1 }^{n_0} \lambda_i | \psi_i \ra \la \psi_i |$ and fix $\ep >0$. According to  the previous steps, we pick up $I\in \N^\ast$ large enough, such that
$\sup_{n \in \N} \sum_{i \ge I }\lambda^n_i+  \sum_{i \ge I } \lambda_i \le \ep$. In the case $n_0 <+\infty$, we may always assume that $I \ge n_0$ and then
\begin{align*}
\| \rho_n - \rho \|_{\SSS_1} & \le \sum_{i \le n_0}
\bigl\|  \lambda^n_i  | \psi^n_i \ra \la \psi^n_i | - \lambda_i  | \psi_i \ra \la \psi_i |\bigr\|_{\SSS_1} + \sum_{i=n_0+1}^{I} \lambda^n_i + \sum_{i > I} \lambda^n_i \\
& \le \sum_{i \le n_0} \bigl(|\lambda^n_i - \lambda_i | + 2 \lambda_i \| \psi_i^n - \psi_i \|_2 \bigr) +  \sum_{i = n_0+1}^{I} \lambda^n_i +  \ep,
\end{align*}
where we have used the fact that 
$ \bigl\|   | \psi^n_i \ra \la \psi^n_i | -  | \psi_i \ra \la \psi_i |\bigr\|_{\SSS_1} \le 2 \| \psi_i - \psi_i^n \|_2$ (see Lemma~\ref{dist_op}). Finally, it is not difficult to see that both sum of the r.h.s go to zero as $n$ goes to infinity, and then that $\rho_n$ converges to $\rho$ in $\SSS_1$.
In the case where $n_0= + \infty$, we may use a similar decomposition to get the same result: replace formally $n_0$ by $I$ is the above bound and erase the second term in the r.h.s.\ .  
\end{proof}


\end{document}